\documentclass[12pt]{article}
\usepackage{fullpage}
\usepackage{amsmath,amssymb,amsfonts,amsthm}
\usepackage{mathtools}
\usepackage{enumerate}
\usepackage{epsfig}
\usepackage[all,poly]{xy}
\usepackage{colortbl}

\usepackage{geometry}
\usepackage{tikz}
\usetikzlibrary{arrows, positioning, shadows}
\usepackage{xcolor}
\usepackage{mathrsfs}

\newif\ifnotesw\noteswtrue% T to show box & marginal notes; F suppresses.
   {\ifnotesw\marginpar[\hfill\(\top\)]{\(\top\)}\fi}%
   {\ifnotesw\marginpar[\hfill\(\bot\)]{\(\bot\)}\fi}

\newcommand{\mnote}[1]%
    {\ifnotesw\marginpar%
        [{\scriptsize\begin{minipage}[t]{\marginparwidth}
        \raggedleft#1%
                        \end{minipage}}]%
        {\scriptsize\begin{minipage}[t]{\marginparwidth}
        \raggedright#1%
                        \end{minipage}}%
    \fi}
\newcommand{\ignore}[1]{}

\newcommand{\etal}{{\it et al.~}}

\newtheorem{theorem}{Theorem}
\newtheorem{corollary}[theorem]{Corollary}
\newtheorem{lemma}[theorem]{Lemma}
\newtheorem{proposition}[theorem]{Proposition}

\newtheorem{definition}{Definition}

\newcommand{\iverson}[1]{\lbrack\!\lbrack #1 \rbrack\!\rbrack}

\newcommand{\zo}{\{0,1\}}

\newcommand{\CC}{\mathbb{C}}
\newcommand{\ZZ}{\mathbb{Z}}
\newcommand{\RR}{\mathbb{R}}
\newcommand{\QQ}{\mathbb{Q}}
\newcommand{\FF}{\mathcal{F}}
\newcommand{\GG}{\mathcal{G}}

\newcommand{\TT}{\mathbb{T}}

\newcommand{\II}{\mathbb{I}}

\newcommand{\anorm}[1]{\lVert #1 \rVert}

\newcommand{\SC}{\mathcal{C}}

\newcommand{\tstar}{t^{\star}}
\newcommand{\vstar}{\varphi^{\star}}

\newcommand{\ket}[1]{\vec{#1}}
\newcommand{\bra}[1]{\vec{#1}^{\dagger}}
\newcommand{\bracket}[2]{\langle {#1} | {#2} \rangle}			% inner product
\newcommand{\tbracket}[3]{\langle {#1} | {#2} {#3} \rangle}		% triple bracket
\newcommand{\braket}[2]{\bra{#1}\ket{#2}}
\newcommand{\ketbra}[2]{\ket{#1}\bra{#2}}

\newcommand{\uket}[1]{\vec{\mathsf{e}}_{#1}}		% for unit vectors
\newcommand{\ubra}[1]{\uket{#1}^{\dagger}}
\newcommand{\buket}[1]{\vec{\mathsf{e}}_{#1}}		% for unit vectors inside inner product (bracket/tbracket)
\newcommand{\bvket}[1]{\vec{#1}}			% for unit-norm vectors inside inner product (bracket/tbracket)

\newcommand{\mbraket}[3]{{#2}_{#1,#3}}			% for matrix access
\newcommand{\mbbraket}[3]{(#2)_{#1,#3}}			% for matrix access (where matrix is in brackets)

\newcommand{\vbbravec}[2]{(\vec{#2})_{#1}}		% for vector access (vector is in brackets)

\renewcommand{\ket}[1]{| #1 \rangle}
\renewcommand{\bra}[1]{\langle #1 |}
\renewcommand{\braket}[2]{\langle #1 | #2 \rangle}
\renewcommand{\ketbra}[2]{| #1 \rangle\langle #2 |}
\renewcommand{\uket}[1]{\ket{#1}}
\renewcommand{\ubra}[1]{\bra{#1}}
\renewcommand{\buket}[1]{#1}

\renewcommand{\bvket}[1]{#1}					
\renewcommand{\mbraket}[3]{\bra{#1}{#2}\ket{#3}}
\renewcommand{\mbbraket}[3]{\bra{#1}#2\ket{#3}}

\renewcommand{\vbbravec}[2]{\braket{#1}{#2}}
\renewcommand{\bracket}[2]{\braket{#1}{#2}}
\renewcommand{\tbracket}[3]{\bra{#1}{#2}\ket{#3}}

\DeclareMathOperator{\SwAut}{\mathsf{SwAut}}
\DeclareMathOperator{\Orb}{Orb}
\DeclareMathOperator{\Fix}{Fix}

\DeclareMathOperator{\id}{\mathit{id}}

\newcommand{\avg}[1]{\langle #1 \rangle}

\DeclareMathOperator{\diag}{diag}

\DeclareMathOperator{\Tr}{Tr}
\DeclareMathOperator{\Circ}{Circ}

\newcommand{\cart}{\mbox{ $\Box$ }}

\newcommand{\Real}{\mathfrak{Re}}
\newcommand{\Imag}{\mathfrak{Im}}

\newcommand{\debug}[1]{{\color{black} #1}}

\title{
Universal State Transfer on Graphs
}
\author{
Stephen Cameron\thanks{Department of Mathematics, College of William and Mary.}
\and
Shannon Fehrenbach\thanks{Department of Mathematics, University of Wisconsin Oshkosh.}
\and
Leah Granger\thanks{Department of Mathematics, Clarkson University.}
\and
Oliver Hennigh\footnotemark[3]
\and
Sunrose Shrestha\thanks{Department of Mathematics, Hamilton College.}
\and
Christino Tamon\thanks{Department of Computer Science, Clarkson University. Contact author: tino@clarkson.edu.}
}
\date{\today}
\begin{document}
\maketitle
\bibliographystyle{plain}

\begin{abstract}
A continuous-time quantum walk on a graph $G$ is given by the unitary matrix 
$U(t) = \exp(-itA)$, where $A$ is the Hermitian adjacency matrix of $G$.
We say $G$ has {\em pretty good state transfer} between vertices $a$ and $b$
if for any $\epsilon > 0$, there is a time $t$, where the $(a,b)$-entry
of $U(t)$ satisfies $|U(t)_{a,b}| \ge 1-\epsilon$.
This notion was introduced by Godsil (2011).
The state transfer is {\em perfect} if the above holds for $\epsilon = 0$.
In this work, we study a natural extension of this notion called
{\em universal} state transfer.
Here, state transfer exists between {\em every} pair of vertices of the graph. 
We prove the following results about graphs with this stronger property:
\begin{itemize}
\item Graphs with universal state transfer have distinct eigenvalues and flat eigenbasis
	(where each eigenvector has entries which are equal in magnitude).
\item The switching automorphism group of a graph with universal state transfer is abelian 
	and its order divides the size of the graph. Moreover, if the state transfer is perfect, 
	then the switching automorphism group is cyclic.
\item There is a family of prime-length cycles with complex weights which has universal 
	pretty good state transfer.
	This provides a concrete example of an infinite family of graphs with the universal property.
\item There exists a class of graphs with real symmetric adjacency matrices which has
	universal {\em pretty good} state transfer. 
	In contrast, Kay (2011) proved that no graph with real-valued adjacency matrix 
	can have universal {\em perfect} state transfer.
\end{itemize}
We also provide a spectral characterization of universal perfect state transfer graphs
that are switching equivalent to circulants.

\vspace{.1in}
\end{abstract}

%%%%%%%%%%%%%%%%%%%%%%%%%%%%%%%%%%%%%%%%%%%%%%%%%%%%%%%%%%%%%%%%%%%%%%%%%%%%%%%%%%%%%%%%%%%%%%%%%%%%%%%%%%%%%%%%%%

\section{Introduction}

The study of continuous-time quantum walk on graphs is important for several reasons.
Originally, it was studied for developing new quantum algorithmic techniques. This led to the
seminal results of Childs \etal \cite{ccdfgs03} and of Farhi \etal \cite{fgg08}.
The algorithms developed in \cite{ccdfgs03,fgg08} are notable in that they provably
beat the corresponding classical resource bounds. The main goal here is to develop 
applications for quantum walk algorithms which have the same impact as Shor's quantum 
algorithm for factoring integers.

Around the same time, Bose \cite{bose03} studied the problem of quantum information
transmission in quantum spin chains. As stated by Bose, this problem may be viewed 
as a continuous-time quantum walk on a path where perfect state transfer occurs
between the two antipodal endpoints. 
Subsequently, Christandl \etal \cite{cdel04,cddekl05} proved that perfect state transfer 
between antipodal points on a path with $n$ vertices is only possible whenever $n=2$ or $3$. 
Despite this negative result, the problem of perfect state transfer on other finite graphs
became an interesting question in both quantum information and algebraic graph theory.
A recent survey in this area is given by Godsil \cite{godsil-dm11}.

More recently, Childs \cite{childs09} showed that quantum walk is a key ingredient 
in simulating universal quantum computation. 
Later, Underwood and Feder \cite{uf10} showed how to use perfect state transfer 
as an alternative to the graph scattering methods used by Childs \cite{childs09}.
This underscores the importance of perfect state transfer in quantum walks.

Several works have raised the difficulties in requiring the state transfer be perfect
(for example, Anderson localization \cite{klmw07}). This led to the natural notion of
{\em pretty good state transfer} which was introduced by Godsil \cite{godsil-dm11}.
In a recent breakthrough, Godsil, Kirkland, Severini and Smith \cite{gkss12} proved 
that some families of paths (whose lengths satisfy specific number-theoretic conditions)
have pretty good state transfer. 
This is in contrast to the negative result of Christandl \etal \cite{cdel04,cddekl05} 
and it confirmed some of the numerical observations made by Bose \cite{bose03}.
A further result on pretty good state transfer was given by Fan and Godsil \cite{fg13}
on the so-called double-star graphs.

A fundamental observation by Kay \cite{kay11} is that perfect state transfer between
two non-disjoint pairs of vertices is impossible for graphs with real symmetric 
adjacency matrices. 
In this work, we study and exhibit graphs with complex Hermitian adjacency matrices 
which have {\em universal} state transfer. 
By universal, we mean that state transfer exists between {\em every} pair of vertices. 
Graphs with complex Hermitian adjacency matrices are also known as complex gain graphs 
-- which are generalizations of signed graphs (see Zaslavsky \cite{z-survey}).
For brevity, we will refer to these graphs as {\em Hermitian} graphs.
As the main result of this work, we prove spectral and structural properties 
of Hermitian graphs with universal state transfer. We summarize these results
in the following.

For spectral properties, we show that if a graph has universal (pretty good or perfect) 
state transfer, then all of its eigenvalues must be simple. 
Moreover, we prove that each eigenvector of such a graph has entries which are all equal 
in magnitude. In other words, the eigenbasis of a graph with universal state transfer 
is {\em flat} 
(and therefore is type-II, since it is also unitary; see Chan and Godsil \cite{cg10}).
For structural properties, we show that the switching automorphism group of a graph 
with universal state transfer must be abelian and its order must divide the size
of the graph. Moreover, if the state transfer is perfect, then the switching automorphism 
group is necessarily cyclic.

Next, we provide an explicit family of graphs with universal state transfer.
For each prime $p$, we prove that the graph $\SC_{p}$ obtained from a directed
$p$-cycle with $\pm i$ weights has universal pretty good state transfer. 
Aside from the smallest case of $\SC_{2}$, the Hermitian $3$-cycle $\SC_{3}$ 
has universal {\em perfect} (not just pretty good) state transfer 
and was the original motivation for our study of universal state transfer.
Our proof employed some the number-theoretic machinery used in 
Godsil \etal \cite{gkss12} (for example, Kronecker's approximation theorem). 
By employing the same methods, we also show that the Cartesian bunkbed
$K_{2} \cart \SC_{p}$ has universal pretty good state transfer for
primes $p \ge 5$.

Given that our example above of prime-length cycles is a circulant family, 
we consider the question of universal state transfer on circulant graphs. 
Here, we prove a spectral characterization of circulants with universal 
perfect state transfer. 
More specifically, the eigenvalues of a $n$-vertex circulant with universal 
perfect state transfer must be a permutation of the integers modulo $n$ under 
a linear bijection. By our earlier result, we know that the switching automorphism
group of a universal perfect state transfer $n$-vertex graph is cyclic and its order 
must divide $n$. But, in the circulant case, we show that the order of this group
must be exactly $n$.

Finally, we give examples of graphs with real adjacency matrices with
universal {\em pretty good} state transfer. This is in contrast to the aforementioned
observation of Kay that universal {\em perfect} state transfer is impossible for such graphs.
Our construction is based on another theorem from number theory (Lindemann's theorem)
and uses Hadamard matrices. Here, we exploit the fact that Hadamard matrices are 
real and flat.

For other works which had studied graphs with complex Hermitian adjacency matrices,
we refer the reader to Kay \cite{kay11} (and some of the references therein) 
and to Zimbor\'{a}s \etal \cite{zfkwlb12}.

%%%%%%%%%%%%%%%%%%%%%%%%%%%%%%%%%%%%%%%%%%%%%%%%%%%%%%%%%%%%%%%%%%%%%%%%%%%%%%%%%%%%%%%%%%%%%%%%%%%%%%%%%%%%%%%%%%

\section{Preliminaries} \label{section:preliminaries}

We state some notation used in the rest of the paper.
For a logical statement $S$, we let the Iversonian $\iverson{S}$ denote 
$1$ if $S$ is true, and $0$ otherwise (see Graham \etal \cite{gkp94}).
As is standard,
$\QQ$ denotes the rational numbers, 
$\ZZ$ denotes the integers, 
$\RR$ denotes the real numbers, 
and
$\CC$ denotes the complex numbers.
Throughout the paper, we reserve $i$ to represent $\sqrt{-1}$.
We also use $\TT$ to denote the set of complex numbers with unit modulus; 
that is, $\TT = \{z \in \CC : |z|=1\}$.

The identity matrix is denoted as $\II$. 
For an index $k$, we let $\uket{k}$ denote the unit vector that is $1$ at position $k$ 
and $0$ elsewhere. The inner product of two vectors $\ket{u}, \ket{v} \in \CC^{n}$ is 
denoted $\bracket{\bvket{u}}{\bvket{v}}$.
All norms used on vectors and matrices will be the standard $2$-norms. 
The conjugate transpose of a matrix $A$ is denoted $A^{\dagger}$.

\subsection{Hermitian Graphs}

We are interested in graphs with Hermitian adjacency matrices.
Let $G$ be a graph over a vertex set $V$ whose adjacency matrix $A$ 
is Hermitian. Here, we view the edge set of $G$ as given
by $E = \{(u,v) \in V \times V : A_{u,v} \neq 0\}$. Since $A$ is Hermitian,
if $(u,v) \in E$ then
\begin{equation}
A_{v,u} = A_{u,v}^{-1} = \overline{A}_{u,v}.
\end{equation}
In this sense, we are working with complex {\em gain} graphs 
(see Zaslavsky \cite{z-survey} and the references therein).
For brevity, we refer to these as {\em Hermitian} graphs 
(since their adjacency matrices are Hermitian). 
In several places, we use the notation $V(G)$, $E(G)$ and $A(G)$
to denote the vertex set, edge set and adjacency matrix of $G$, respectively.

A {\em monomial} $n \times n$ matrix is a product of a permutation matrix $P_{\phi}$,
where $\phi$ is a permutation on $n$ elements, and a complex diagonal matrix $D$ 
of size $n$ (see Davis \cite{davis}). 
The $(j,k)$-entry of the permutation matrix $P_{\phi}$ is given by
%$\bra{j}P_{\phi}\ket{k} = \iverson{j = \phi(k)}$, 
$(P_{\phi})_{j,k} = \iverson{j = \phi(k)}$, 
for each $j,k \in \{1,\ldots,n\}$.
We use the notation $\tilde{P}_{\phi} = P_{\phi}D$ to denote such a monomial matrix,
where we intentionally suppress details about $D$ since our focus is on the permutation 
action of $\phi$ (rather than the scaling by $D$).
Any monomial matrix is a unitary matrix since
\begin{equation}
(P_{\phi}D)^{\dagger} 
	= \overline{D}P_{\phi}^{T}
	= (P_{\phi}D)^{-1}.
\end{equation}
Thus, the eigenvalues of a monomial matrix all lie in $\TT$.
Also, note that monomial matrices are closed under multiplication since
\begin{equation}
\tilde{P}_{\phi_{1}}\tilde{P}_{\phi_{2}}
=
(P_{\phi_{1}}D_{1})(P_{\phi_{2}}D_{2})
=
P_{\phi_{1}}P_{\phi_{2}}\hat{D}_{1}D_{2}
=
\tilde{P}_{\phi_{1}\phi_{2}},
\end{equation}
%where $\bra{j}\hat{D}_{1}\ket{j} = \bra{\phi_{2}(j)}D_{1}\ket{\phi_{2}(j)}$
where the $(j,j)$-entry of $\hat{D}_{1}$ is equal to the 
$(\phi_{2}(j),\phi_{2}(j))$-entry of $D_{1}$.
Here, we use $\phi_{1}\phi_{2}$ to denote the composition of the permutations
$\phi_{1}$ and $\phi_{2}$.
Since $\II$ is monomial, the set of monomial matrices form a group under
multiplication.

We say two Hermitian graphs $G_{1}$ and $G_{2}$ are {\em switching isomorphic}, 
or $G_{1} \simeq G_{2}$, if there is a monomial matrix $\tilde{P}_{\phi}$,
where $\phi$ is a bijection $\phi: V(G_{2}) \rightarrow V(G_{1})$, with
\begin{equation}
A(G_{2}) 
= \tilde{P}_{\phi}^{\dagger}A(G_{1})\tilde{P}_{\phi}.
= (P_{\phi}D)^{\dagger}A(G_{1})(P_{\phi}D).
\end{equation}
We say $G_{1}$ and $G_{2}$ are {\em switching equivalent}, or $G_{1} \sim G_{2}$, 
if they are switching isomorphic with $P_{\phi} = \II$.
We say $G_{1}$ and $G_{2}$ are {\em isomorphic}, or $G_{1} \cong G_{2}$, 
if they are switching isomorphic with $D = \II$.
The switching automorphism group $\SwAut(G)$ is the group of monomial matrices 
which commute with $A(G)$:
\begin{equation}
\SwAut(G) 
= 
\{\tilde{P}_{\phi}  : A(G)\tilde{P}_{\phi} = \tilde{P}_{\phi}A(G)\}.
\end{equation}
Note that if $P_{\phi}D_{1}$ and $P_{\phi}D_{2}$ are both in $\SwAut(G)$,
then $D_{1} = \gamma D_{2}$, for some $\gamma \in \CC$.

\subsection{Circulants}

A circulant graph is a graph whose adjacency matrix is circulant.
An $n \times n$ circulant matrix $C = \Circ(a_{0},\ldots,a_{n-1})$ defined by 
the sequence $\{a_{k}\}$ is given by
\begin{equation} \label{eqn:circulant-matrix}
C =
\begin{bmatrix}
a_{0} & a_{1} & a_{2} & \ldots & a_{n-1} \\
a_{n-1} & a_{0} & a_{1} & \ldots & a_{n-2} \\
a_{n-2} & a_{n-1} & a_{0} & \ldots & a_{n-3} \\
\vdots & \vdots & \vdots & \ldots & \vdots \\
a_{1} & a_{2} & a_{3} & \ldots & a_{0}
\end{bmatrix}.
\end{equation}
Let $\sigma = (0 \ 1 \ \ldots \ n-1)$ denote the $n$-cycle permutation whose 
permutation matrix is given by the circulant matrix $\Theta_{n} = \Circ(0,0,\ldots,0,1)$:
%\ignore{
\begin{equation}
\Theta_{n} =
\begin{bmatrix}
0 & 0 & 0 & \ldots & 0 & 1 \\
1 & 0 & 0 & \ldots & 0 & 0 \\
0 & 1 & 0 & \ldots & 0 & 0 \\
\vdots & \vdots & \vdots & \vdots & \vdots & \vdots \\
0 & 0 & 0 & \ldots & 1 & 0
\end{bmatrix}
\end{equation}
%}
An alternate and useful characterization of circulants is that they are polynomials
in $\Theta_{n}$. For example, the circulant $C$ in Equation (\ref{eqn:circulant-matrix}) may
be written as 
\begin{equation}
C = \sum_{k=0}^{n-1} a_{n-k}\Theta_{n}^{k}.
\end{equation}
Let $\omega_{n} = \exp(2\pi i/n)$ be the principal $n$th root of unity.
The Fourier matrix $\FF_{n}$ of order $n$ is a unitary matrix
whose entries are given by
$\mbbraket{j}{\FF_{n}}{k} = \omega_{n}^{jk}/\sqrt{n}$,
for each $j,k \in \{0,1,\ldots,n-1\}$.
For convenience, we often refer to the $k$th column of $\FF_{n}$ as $\ket{F_{k}}$.
It is known that $\FF$ diagonalizes any circulant $C$; that is, we have
\begin{equation} \label{eqn:fourier-diagonalize-circulant}
\FF_{n}^{\dagger} C \FF_{n} = 
\begin{bmatrix}
\lambda_{0} & 0 & 0 & \ldots & 0 & 0 \\
0 & \lambda_{1} & 0 & \ldots & 0 & 0 \\
0 & 0 & \lambda_{2} & \ldots & 0 & 0 \\
\vdots & \vdots & \vdots & \vdots & \vdots & \vdots \\
0 & 0 & 0 & \ldots & 0 & \lambda_{n-1}
\end{bmatrix}
\end{equation}
By using Equation (\ref{eqn:fourier-diagonalize-circulant}), we may derive the eigenvalues of 
the circulant $C$ which are given by
\begin{equation} \label{eqn:circulant-eigenvalue}
\lambda_{k} = \sum_{j=0}^{n-1} a_{j} \omega_{n}^{jk}.
\end{equation}
More background on circulants may be found in Davis \cite{davis} and 
in Godsil and Royle \cite{godsil-royle01}.

\subsection{Quantum walk}

Given a graph $G=(V,E)$ with adjacency matrix $A$, a continuous-time quantum walk on $G$
is defined by the time-dependent unitary matrix $U(t) = e^{-itA}$, where $t \in \RR$.
For two vertices $a,b \in V$, we say that $G$ has {\em pretty good state transfer} from
$a$ to $b$ if for any $\epsilon > 0$, there is a time $t$ so that
\begin{equation} \label{eqn:pgst}
\anorm{e^{-itA}\uket{a} - \gamma\uket{b}} < \epsilon,
\end{equation}
for some $\gamma \in \TT$.
When the context is clear, we will use the convenient shorthand 
$e^{-itA}\uket{a} \approx \gamma\uket{b}$ to denote Equation (\ref{eqn:pgst}).
\debug{Alternatively}, we may require 
$|\tbracket{\buket{b}}{e^{-itA}}{\buket{a}}| \ge 1 - \epsilon$
in place of Equation (\ref{eqn:pgst}).

We say $G$ has {\em perfect state transfer} from vertex $a$ to vertex $b$ at time $t$ if 
\begin{equation} \label{eqn:pst}
|\tbracket{\buket{b}}{e^{-itA}}{\buket{a}}| = 1.
\end{equation}
A graph $G$ is called {\em periodic} if there is perfect state transfer from $a$ to itself, 
for every vertex $a$ of $G$.

\begin{definition}
A graph $G$ has {\em universal} (pretty good or perfect) {\em state transfer} 
if it has (pretty good or perfect) state transfer between every pair of its vertices.
\end{definition}

%%%%%%%%%%%%%%%%%%%%%%%%%%%%%%%%%%%%%%%%%%%%%%%%%%%%%%%%%%%%%%%%%%%%%%%%%%%%%%%%%%%%%%%%%%%%%%%%%%%%%%%%%%%%%%%%%%

\section{Small Graphs with Universal State Transfer} \label{section:small-graphs}

Here, we provide small examples of graphs with universal (perfect and pretty good)
state transfer. First, note that any Hermitian graph based on $K_{2}$ has universal perfect 
state transfer. We may take, for example, graphs whose adjacency matrices are the standard 
Pauli $X$ and $Y$ matrices:
\begin{equation}
X = \begin{bmatrix} 0 & 1 \\ 1 & 0 \end{bmatrix},
\ 
\hspace{.5in}
\
Y = \begin{bmatrix} 0 & -i \\ i & 0 \end{bmatrix}.
\end{equation}
If $\mathcal{K}_{2}$ is a Hermitian graph with adjacency matrix $Y$,
the quantum walk on this graph is
\begin{equation}
e^{-itY}
=
\frac{e^{-it}}{2}
\begin{bmatrix}
1 & -i \\
i & 1
\end{bmatrix}
+
\frac{e^{it}}{2}
\begin{bmatrix}
1 & i \\
-i & 1
\end{bmatrix}
=
\begin{bmatrix}
\cos(t) & -\sin(t) \\
\sin(t) & \cos(t)
\end{bmatrix}
\end{equation}
which shows that it has universal perfect state transfer,
since each matrix entry can be made equal to $1$.

\subsection{Hermitian $3$-Cycle}

Consider the Hermitian graph $\SC_{3}$ (see Figure \ref{figure:c3-k4})
whose adjacency matrix is
\begin{equation}
A(\SC_{3})
=
\begin{bmatrix}
0 & -i & i \\
i & 0 & -i \\
-i & i & 0
\end{bmatrix}.
\end{equation}
Let $\omega = e^{2\pi i/3}$.
Since $\SC_{3}$ is a circulant, its eigenvalues are
$\lambda_{0} = 0$, $\lambda_{1} = \sqrt{3}$ and $\lambda_{2} = -\sqrt{3}$.
Thus, by the spectral decomposition, we have
\begin{equation}
e^{-itA(\SC_{3})}
=
\frac{1}{3}
\begin{bmatrix}
1 & 1 & 1 \\
1 & 1 & 1 \\
1 & 1 & 1
\end{bmatrix}
+
\frac{e^{-it\sqrt{3}}}{3}
\begin{bmatrix}
1 & \overline{\omega} & \omega \\
\omega & 1 & \overline{\omega} \\
\overline{\omega} & \omega & 1
\end{bmatrix}
+
\frac{e^{it\sqrt{3}}}{3}
\begin{bmatrix}
1 & \omega & \overline{\omega} \\
\overline{\omega} & 1 & \omega \\
\omega & \overline{\omega} & 1
\end{bmatrix}
\end{equation}
The quantum walk starting at $0$ (without loss of generality) is given by
\begin{equation}
e^{-itA(\SC_{3})}\uket{0}
	= 
		\frac{1}{3}
		\begin{bmatrix}
		1 + e^{-it\sqrt{3}} + e^{it\sqrt{3}} \\
		1 + \omega e^{-it\sqrt{3}} + \overline{\omega} e^{it\sqrt{3}} \\
		1 + \overline{\omega} e^{-it\sqrt{3}} + \omega e^{it\sqrt{3}} 
		\end{bmatrix} 
	= 
		\frac{1}{3}
		\begin{bmatrix}
		1 + 2\cos(t\sqrt{3}) \\
		1 + 2\cos(t\sqrt{3} - 2\pi/3) \\
		1 + 2\cos(t\sqrt{3} + 2\pi/3)
		\end{bmatrix}
\end{equation}
Thus, we have perfect state transfer 
from $0$ to $0$ at time $t = 2\pi/\sqrt{3}$,
from $0$ to $1$ at time $t = 8\pi/(3\sqrt{3})$,
and
from $0$ to $2$ at time $t = 4\pi/(3\sqrt{3})$.

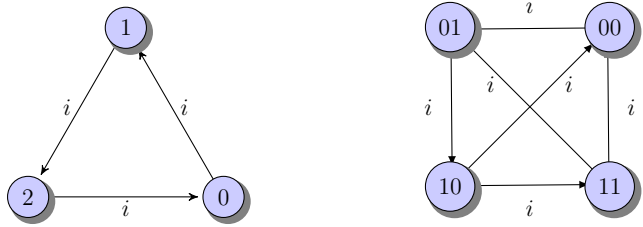
\begin{figure}[t]
\begin{center}
\begin{tikzpicture}[->, >=stealth', shorten >=2pt, scale=1.5]%K_3 complex
\foreach \c in {0,...,2}
{
	%\node at ({\c*90+(\c-1)*30}:1)[circle, circular drop shadow, fill={blue!20}, draw][scale=0.75]{$\c$};
	\draw ({\c*90+(\c-1)*30+6}:0.85) -> ({(\c+1)*90+\c*30-6}:0.85);
	%\draw [->]({(\c+1)*90+\c*30-10}:1) arc ({(\c+1)*90+\c*30-10}:{\c*90+(\c-1)*30+10}:1);
	\node at ({\c*120+30}:0.6)[circle][scale=0.75]{$i$};
	%\node at ({\c*120+30}:1.2)[circle][scale=0.75]{$-i$};
}
\foreach \c in {0,...,2}
{
	\node at ({\c*90+(\c-1)*30}:1)[circle, circular drop shadow, fill={blue!20}, draw][scale=0.75]{$\c$};
}
\end{tikzpicture}
\quad \quad 
\hspace{.1in}
\quad \quad
\begin{tikzpicture}[scale=1.5][h!!!!!!!!!]
\foreach \c [count=\x from 0] in {0, ..., 3} %list for label, count for degree
{
	%\node at (\x*90+45:1)[circle, circular drop shadow, fill={blue!20}, draw][scale=0.75]{$\c$};	
	\draw [->, >=latex, shorten >= 2pt]({\x*90+45+6}:0.9) -> ({(\x+1)*90+45-5}:0.9); %straight edge
	\node at (\x*90:0.9)[circle][scale=0.75]{$i$}; %edge label	
}
\draw [->, >=latex, shorten >= 2pt](225:0.8) -> (45:0.85);	
\draw [->, >=latex, shorten >= 2pt](315:0.8) -> (135:0.85);

% switch off integer labels
%\foreach \c [count=\x from 0] in {0, ..., 3} %list for label, count for degree
%{
%	\node at (\x*90+45:1)[circle, circular drop shadow, fill={blue!20}, draw][scale=0.75]{$\c$};	
%}
% switch on 2-cube labels
{
	\node at (0*90+45:1)[circle, circular drop shadow, fill={blue!20}, draw][scale=0.75]{$00$};
	\node at (1*90+45:1)[circle, circular drop shadow, fill={blue!20}, draw][scale=0.75]{$01$};
	\node at (2*90+45:1)[circle, circular drop shadow, fill={blue!20}, draw][scale=0.75]{$10$};
	\node at (3*90+45:1)[circle, circular drop shadow, fill={blue!20}, draw][scale=0.75]{$11$};
}
\node at (30: 0.4)[circle][scale=0.75]{$i$};
\node at (150: 0.4)[circle][scale=0.75]{$i$};
\end{tikzpicture}
\caption{
Hermitian graphs:
(1) $\SC_{3}$ has universal perfect state transfer;
(2) $\mathcal{K}_{4}$ has universal pretty good state transfer.
For each directed edge labeled with $i$, there is a backward edge labeled with $-i$
(which are omitted for simplicity).
}
\label{figure:c3-k4}
\end{center}
\end{figure}

\subsection{Hermitian $4$-Clique}

Our second example is the Hermitian graph $\mathcal{K}_{4}$ (see Figure \ref{figure:c3-k4})
which is a graph on the vertex set $\zo^{2} = \{00,01,10,11\}$.
The adjacency matrix of this graph is given by:
\begin{equation}
A(\mathcal{K}_{4})
=
\begin{bmatrix}
 0 & -i & i & i \\
i &  0 & -i & i \\
-i & i &  0 & -i \\
-i & -i & i &  0
\end{bmatrix}
\ = \
\II \otimes Y - (Y \otimes \II + X \otimes Y).
\end{equation}
A switching automorphisms of $\mathcal{K}_{4}$ is given by the following
monomial matrix:
\begin{equation}
\tilde{P}_{(0 1)(2 3)}
=
\begin{bmatrix}
0  & i & 0  & 0  \\
-i & 0 & 0  & 0 \\
0  & 0 & 0  & i \\
0  & 0 & -i & 0 
\end{bmatrix}
=
\begin{bmatrix}
0 & 1 & 0 & 0 \\
1 & 0 & 0 & 0 \\
0 & 0 & 0 & 1 \\
0 & 0 & 1 & 0
\end{bmatrix}
\begin{bmatrix}
-i & 0 & 0 & 0 \\
0  & i & 0 & 0 \\
0  & 0 & -i & 0 \\
0  & 0 & 0 & i 
\end{bmatrix}
\end{equation}
To analyze the quantum walk on $\mathcal{K}_{4}$, we need the following lemma
which might be of independent interest.

\begin{lemma} \label{lemma:anticommute-exponential}
Let $A$ and $B$ be two anti-commuting square matrices, that is $AB = -BA$,
which satisfies $A^{2} + B^{2} > 0$. Then, we have
\begin{equation}
e^{it(A+B)} = 
\cos(t\sqrt{A^{2} + B^{2}}) + i\frac{A+B}{\sqrt{A^{2} + B^{2}}} \sin(t\sqrt{A^{2} + B^{2}}).
\end{equation}
\end{lemma}
\begin{proof}
Since $AB + BA = 0$, we have $A^{2}$ and $B^{2}$ commute.
This yields
\begin{eqnarray}
(A + B)^{2m} & = & (A^{2} + B^{2})^{m}  \\
(A + B)^{2m+1} & = & (A + B)(A^{2} + B^{2})^{m}.
\end{eqnarray}
Therefore, the exponential matrix $e^{it(A+B)}$ is given by
\begin{eqnarray}
e^{it(A+B)}
	& = & 
		\sum_{m=0}^{\infty} \frac{(it)^{2m}}{(2m)!} (A^{2}+B^{2})^{m}
		+
		(A+B) \sum_{m=0}^{\infty} \frac{(it)^{2m+1}}{(2m+1)!} (A^{2}+B^{2})^{m} \\
	& = &
		\cos(t \sqrt{A^{2}+B^{2}})
		+
		i\frac{(A+B)}{\sqrt{A^{2}+B^{2}}}
		\sin(t \sqrt{A^{2}+B^{2}})
\end{eqnarray}
This proves the claim.
\end{proof}

\par\noindent
We apply Lemma \ref{lemma:anticommute-exponential} with
$A = Y \otimes \II$ and $B = X \otimes Y$. 
Since $A^{2} + B^{2} = 2\II$, we have
\begin{eqnarray} \label{eqn:k4-fidelity}
e^{it(A+B)}
& = & \cos(t\sqrt{2})\II + \frac{1}{\sqrt{2}}\sin(t\sqrt{2})(A+B) \\
& = &
\frac{1}{\sqrt{2}}
\begin{bmatrix}
\sqrt{2}\cos(t\sqrt{2}) & 0 & \sin(t\sqrt{2}) & \sin(t\sqrt{2}) \\
0 & \sqrt{2}\cos(t\sqrt{2}) & -\sin(t\sqrt{2}) & \sin(t\sqrt{2}) \\
-\sin(t\sqrt{2}) & \sin(t\sqrt{2}) & \sqrt{2}\cos(t\sqrt{2}) & 0 \\
-\sin(t\sqrt{2}) & -\sin(t\sqrt{2}) & 0 & \sqrt{2}\cos(t\sqrt{2}) 
\end{bmatrix}
\end{eqnarray}
Also, note that $\II \otimes Y$ commutes with both $A$ and $B$.
Given the tensor decomposition of $A(\mathcal{K}_{4})$, we will view the 
vertices of $\mathcal{K}_{4}$ as binary strings over $\zo^{2}$.
The quantum walk on $\mathcal{K}_{4}$ starting at vertex $00$ is given by
\begin{eqnarray}
e^{-itA(\mathcal{K}_{4})}\uket{00}
	& = & e^{it(A+B)} e^{-it(\II \otimes Y)}\uket{00} 
		= e^{it(A+B)} (\II \otimes e^{-itY})\uket{00} \\
	& = & e^{it(A+B)} (\cos(t)\uket{00} + \sin(t)\uket{01}).
\end{eqnarray}
Using Equation (\ref{eqn:k4-fidelity}), we see that
\begin{eqnarray}
e^{-itA(\mathcal{K}_{4})}\uket{00}
	& = & \cos(t)\cos(t\sqrt{2}) \uket{00} + \sin(t)\cos(t\sqrt{2}) \uket{01} \\
	& & + \ \mbox{\normalsize $\frac{1}{\sqrt{2}}$} (-\cos(t) + \sin(t)) \sin(t\sqrt{2}) \uket{10} \\ 
	& & + \ \mbox{\normalsize $\frac{1}{\sqrt{2}}$} (-\cos(t) - \sin(t)) \sin(t\sqrt{2}) \uket{11}
\end{eqnarray}
Since $\{1,\sqrt{2}\}$ is linearly independent over the rationals,
$\mathcal{K}_{4}$ has perfect state transfer from vertex $00$ to all other vertices.
For example, to reach vertex $11$, we can arrange so that 
$\sin(t) \approx 1/\sqrt{2}$, $\cos(t) \approx 1/\sqrt{2}$ 
and $\sin(t\sqrt{2}) \approx 1$ simultaneously for some $t$.
We will formalize this argument using Kronecker's Theorem 
for several family of graphs in Theorem \ref{thm:circulant-with-upgst}.

%%%%%%%%%%%%%%%%%%%%%%%%%%%%%%%%%%%%%%%%%%%%%%%%%%%%%%%%%%%%%%%%%%%%%%%%%%%%%%%%%%%%%%%%%%%%%%%%%%%%%%%%%%%%%%%%%%

\section{State Transfer on Graphs}

In this section, we prove that if a graph $G$ has state transfer from 
a vertex $a$ to a vertex $\phi(a)$ for some automorphism $\phi$ of $G$, 
then the unitary matrix $e^{-itA(G)}$ is equal to the permutation matrix $P_{\phi}$
up to scaling.
This property was first observed by Godsil \etal \cite{gkss12} in the context of 
pretty good state transfer on paths. We generalize their observation to arbitrary 
graphs under some mild assumptions. First, we prove this for perfect state transfer
since the proof is simpler and more direct.

\begin{theorem} \label{thm:pst-implies-permutation}
Let $G$ be a $n$-vertex graph with {perfect} state transfer 
from vertex $a$ to vertex $\phi(a)$ at time $t$,
for some switching automorphism $\tilde{P}_{\phi} \in \SwAut(G)$.
Suppose that $A(G)$ and $\tilde{P}_{\phi}$ share a set 
$\{\ket{z_{1}},\ldots,\ket{z_{n}}\}$ of orthonormal eigenvectors
which satisfies $\bracket{\buket{a}}{\bvket{z_{k}}} \neq 0$, for each $k$.
Then, 
\begin{equation}
e^{-itA(G)} = \gamma\tilde{P}_{\phi},
\end{equation}
for some $\gamma \in \TT$.
\end{theorem}
\begin{proof}
Let $\lambda_{1},\ldots,\lambda_{n}$ be the eigenvalues of $A(G)$.
Given that $\tilde{P}_{\phi}$ is a monomial matrix, we assume its eigenvalues 
are given by $e^{\theta_{1}},\ldots,e^{\theta_{n}}$, for some real numbers $\theta_{k} \in \RR$.
Since $A(G)$ and $\tilde{P}_{\phi}$ share the set $\{\ket{z_{1}},\ldots,\ket{z_{n}}\}$ 
of orthonormal eigenvectors, we have
\begin{equation}
A(G) = \sum_{k=1}^{n} \lambda_{k}\ketbra{z_{k}}{z_{k}},
\ \ \
\hspace{.2in}
\ \ \
\tilde{P}_{\phi} = \sum_{k=1}^{n} e^{i\theta_{k}}\ketbra{z_{k}}{z_{k}}.
\end{equation}
Since $G$ has perfect state transfer from vertex $a$ to vertex $\phi(a)$ at time $t$, we have
\begin{equation}
1 = |\tbracket{\buket{a}}{\tilde{P}^{\dagger}_{\phi}e^{-itA(G)}}{\buket{a}}| 
	= \left|\sum_{k=1}^{n} e^{-i(t\lambda_{k}+\theta_{k})} 
		\bracket{\buket{a}}{\bvket{z_{k}}} \bracket{\bvket{z_{k}}}{\buket{a}}\right| 
	\le \sum_{k=1}^{n} |\bracket{\buket{a}}{\bvket{z_{k}}}|^{2}
	= 1,
\end{equation}
where the last equality holds since $\sum_{k} \ketbra{z_{k}}{z_{k}} = \II$.
Thus, there is a real number $\alpha \in \RR$ so that
$e^{-i(t\lambda_{k} + \theta_{k})} = e^{i\alpha}$,
for each $k=1,\ldots,n$ (see Lemma \ref{lemma:complex-sum-forcing}).
Therefore, we have
\begin{equation}
e^{-itA(G)} 
	= \sum_{k=1}^{n} e^{-it\lambda_{k}} \ketbra{z_{k}}{z_{k}}
	= e^{i\alpha} \sum_{k=1}^{n} e^{i\theta_{k}} \ketbra{z_{k}}{z_{k}}
	= e^{i\alpha} \tilde{P}_{\phi}.
\end{equation}
This proves the claim.
\end{proof}

Next, we generalize Theorem \ref{thm:pst-implies-permutation} for the case of
pretty good state transfer.

\begin{theorem} \label{thm:pgst-implies-permutation}
Let $G$ be a $n$-vertex graph with pretty good state transfer 
from vertex $a$ to vertex $\phi(a)$,
for some switching automorphism $\tilde{P}_{\phi} \in \SwAut(G)$.
Suppose that $A(G)$ and $\tilde{P}_{\phi}$ share a set 
$\{\ket{z_{1}},\ldots,\ket{z_{n}}\}$ of orthonormal eigenvectors
which satisfies $\bracket{\buket{a}}{\bvket{z_{k}}} \neq 0$, for each $k$.
Then, for each $\epsilon > 0$ there is a time $t \in \RR$ where
\begin{equation}
\anorm{e^{-itA(G)} - \gamma\tilde{P}_{\phi}} \le \epsilon,
\end{equation}
for some $\gamma \in \TT$.
\end{theorem}
\begin{proof}
Suppose $A(G)$ has eigenvalues $\lambda_{1},\ldots,\lambda_{n}$.
Given that $\tilde{P}_{\phi}$ is a monomial matrix, 
we assume its eigenvalues are given by $e^{\theta_{1}},\ldots,e^{\theta_{n}}$, 
for some real numbers $\theta_{k} \in \RR$.
Since $A(G)$ and $\tilde{P}_{\phi}$ share the set $\{\ket{z_{1}},\ldots,\ket{z_{n}}\}$ 
of orthonormal eigenvectors, we have
$A(G) = \sum_{k=1}^{n} \lambda_{k}\ketbra{z_{k}}{z_{k}}$
and
$\tilde{P}_{\phi} = \sum_{k=1}^{n} e^{i\theta_{k}}\ketbra{z_{k}}{z_{k}}$.
Thus, we have
\begin{equation} \label{eqn:pgst-decomp}
\tilde{P}_{\phi}^{\dagger}e^{-i\tilde{t}A(G)}
=
\sum_{k=1}^{n} e^{-i(\tilde{t}\lambda_{k}+\theta_{k})}\ketbra{z_{k}}{z_{k}}.
\end{equation}
Since $G$ has pretty good state transfer from vertex $a$ to vertex $\phi(a)$,
for any $\tilde{\epsilon} > 0$, there is a time $\tilde{t}$ where
\begin{equation} \label{eqn:pgst-vector}
\anorm{e^{-i\tilde{t}A(G)}\uket{a} - \gamma\tilde{P}_{\phi}\uket{a}} \ < \ \tilde{\epsilon},
\end{equation}
for some $\gamma \in \TT$. 
By using Equation (\ref{eqn:pgst-decomp}) and the fact that $\tilde{P}_{\phi}$ is unitary, 
we rewrite the left-hand side of Equation (\ref{eqn:pgst-vector}) as
\begin{equation} \label{eqn:pgst-spectral}
\anorm{e^{-i\tilde{t}A(G)}\uket{a} - \gamma\tilde{P}_{\phi}\uket{a}}
	=
	\anorm{(\tilde{P}_{\phi}^{\dagger}e^{-i\tilde{t}A(G)} - \gamma \II)\uket{a}} 
	=
	\left\lVert \sum_{k=1}^{n} (e^{-i(\tilde{t}\lambda_{k}+\theta_{k})}-\gamma)
		\bracket{\bvket{z_{k}}}{\buket{a}} \ket{z_{k}} \right\rVert.
\end{equation}
Since $\{\ket{z_{1}},\ldots,\ket{z_{n}}\}$ is an orthonormal set, 
by combining Equations (\ref{eqn:pgst-vector}) and (\ref{eqn:pgst-spectral}), we have
\begin{equation}
\left\lVert \sum_{k=1}^{n} 
	(e^{-i(\tilde{t}\lambda_{k}+\theta_{k})}-\gamma) 
		\bracket{\bvket{z_{k}}}{\buket{a}} \ket{z_{k}} \right\rVert^{2} 
\ = \
\sum_{k=1}^{n} |e^{-i(\tilde{t}\lambda_{k}+\theta_{k})}-\gamma|^{2} 
	\cdot
	|\bracket{\buket{a}}{\bvket{z_{k}}}|^{2} 
\ < \ \tilde{\epsilon}^{2}.
\end{equation}
Suppose that $\delta = \min_{k}\{|\bracket{\buket{a}}{\bvket{z_{k}}}|\}$, 
where $0 < \delta \le 1$.
Then, we get
\begin{equation} \label{eqn:cauchy-schwarz}
\sum_{k=1}^{n} \left|e^{-i(\tilde{t}\lambda_{k}+\theta_{k})}-\gamma\right|^{2}
\ < \ 
\left(\frac{\tilde{\epsilon}}{\delta}\right)^{2}.
\end{equation}
This allows us to obtain the following upper bound:
\begin{eqnarray}
\anorm{\tilde{P}_{\phi}^{\dagger}e^{-i\tilde{t}A(G)} - \gamma\II}
	& = & \left\lVert \sum_{k=1}^{n} (e^{-i(\tilde{t}\lambda_{k}+\theta_{k})} 
			- \gamma)\ketbra{z_{k}}{z_{k}} \right\rVert \\
%		\ \ \mbox{ by spectral decomposition } \\
	& \le & \sum_{k=1}^{n} |e^{-i(\tilde{t}\lambda_{k}+\theta_{k})} - \gamma| 
			\cdot \anorm{\ketbra{z_{k}}{z_{k}}} \\
%		\ \ \mbox{ by triangle inequality } \\
	& \le & \sum_{k=1}^{n} |e^{-i(\tilde{t}\lambda_{k}+\theta_{k})} - \gamma|,
		\ \ \mbox{ since $\anorm{\ketbra{z_{k}}{z_{k}}} \le 1$ } \\
	& \le & \sqrt{n} \left(\frac{\tilde{\epsilon}}{\delta}\right), 
		\ \ \mbox{ by Cauchy-Schwarz on Equation (\ref{eqn:cauchy-schwarz}).}
\end{eqnarray}
We have used
$\anorm{\ketbra{z_{k}}{z_{k}}} = \max \ \anorm{(\ket{z_{k}}\bra{z_{k}})\ket{x}}
\le |\bracket{\bvket{z_{k}}}{\bvket{x}}| \le 1$, where the maximum is taken over all unit-norm
$\ket{x}$.

So, given $\epsilon > 0$, we choose $\tilde{\epsilon} = \delta\epsilon/\sqrt{n}$.
Since $G$ has pretty good state transfer from $a$ to $\phi(a)$, there is a time $\tilde{t}$
so that $\anorm{e^{-i\tilde{t}A(G)}\uket{a} - \gamma\tilde{P}_{\phi}\uket{a}} 
< \delta\epsilon/\sqrt{n} = \tilde{\epsilon}$.
From the derivation above, 
we get
\begin{equation}
\anorm{e^{-i\tilde{t}A(G)} - \gamma\tilde{P}_{\phi}} 
\ = \
\anorm{\tilde{P}_{\phi}^{\dagger}e^{-i\tilde{t}A(G)} - \gamma\II}
\ < \
\epsilon.
\end{equation}
This proves the claim.
\end{proof}

In \cite{godsil-dm11}, Godsil proved that if $G$ is a vertex-transitive graph
where perfect state transfer occurs between vertices $a$ and $b$ at time $t$,
then $e^{-itA(G)}$ is a scalar multiple of a permutation matrix $P$, 
for which $P\uket{a} = \uket{b}$, where $P$ is of order two with no fixed points 
and $P$ lies in the center of the automorphism group of $G$.
We generalize this observation to Hermitian graphs.

\begin{corollary} \label{cor:pst-implies-center}
Suppose $G$ has perfect state transfer from vertex $a$ to vertex $\phi(a)$ at time $t$,
for some switching automorphism $\tilde{P}_{\phi} \in \SwAut(G)$.
Assume that $A(G)$ and $\tilde{P}_{\phi}$ share a set $\{\ket{z_{1}},\ldots,\ket{z_{n}}\}$
of orthonormal eigenvectors, where $\bracket{\buket{a}}{\bvket{z_{k}}} \neq 0$, for each $k$.
Then, $G$ has perfect state transfer from vertex $b$ to vertex $\phi(b)$ at time $t$, 
for each $b \in V(G)$.
Moreover, $\tilde{P}_{\phi}$ is in the center of $\SwAut(G)$.
\end{corollary}
\begin{proof}
Suppose $G$ has perfect state transfer from $a$ to $\phi(a)$ at time $t$, for some
$\tilde{P}_{\phi} \in \SwAut(G)$. 
Assume $A(G)$ and $\tilde{P}_{\phi}$ share an eigenbasis for which vertex $a$ 
has full eigenvector support.
By Theorem \ref{thm:pst-implies-permutation}, we have 
$e^{-itA(G)} = \gamma\tilde{P}_{\phi}$, for some $\gamma \in \TT$. 
So, for each vertex $b$, we have
\begin{equation}
e^{-itA(G)}\uket{b} \ = \ (\gamma\tilde{P}_{\phi}) \uket{b} \ = \ \gamma \uket{\phi(b)}.
\end{equation}
Suppose $A(G) = \sum_{k=} \lambda_{k}E_{k}$, where $E_{k}$ is the eigenprojector 
corresponding to eigenvalue $\lambda_{k}$. Thus, we have
\begin{equation}
e^{-itA(G)} = \sum_{k} e^{-it\lambda_{k}}E_{k}.
\end{equation}
Since each eigenprojector $E_{k}$ is a polynomial in $A(G)$ (see Godsil \cite{godsil93}),
every switching automorphism $\tilde{P}_{\psi} \in \SwAut(G)$ commutes with $e^{-itA(G)}$
(by definition $\tilde{P}_{\psi}$ commutes with $A(G)$). 
Since $\tilde{P}_{\phi}$ is a scalar multiple of $e^{-itA(G)}$, 
each $\tilde{P}_{\psi}$ commutes with $\tilde{P}_{\phi}$.
\end{proof}

Next, we prove the version of Corollary \ref{cor:pst-implies-center} for
pretty good state transfer.

\begin{corollary} \label{cor:pgst-implies-center}
Suppose $G$ has pretty good state transfer from vertex $a$ to vertex $\phi(a)$,
for some switching automorphism $\tilde{P}_{\phi} \in \SwAut(G)$.
Assume $A(G)$ and $\tilde{P}_{\phi}$ share a set $\{\ket{z_{1}},\ldots,\ket{z_{n}}\}$
of orthonormal eigenvectors, where $\bracket{\buket{a}}{\bvket{z_{k}}} \neq 0$, for each $k$.
Then, $G$ has pretty good state transfer from vertex $b$ to vertex $\phi(b)$, 
for each $b \in V(G)$.
Moreover, $\tilde{P}_{\phi}$ is in the center of $\SwAut(G)$.
\end{corollary}
\begin{proof}
Let $G$ be a graph with pretty good state transfer from $a$ to $\phi(a)$,
for some $\tilde{P}_{\phi} \in \SwAut(G)$.
Assume that $A(G)$ and $\tilde{P}_{\phi}$ share an eigenbasis where $a$ has full support.
By Theorem \ref{thm:pgst-implies-permutation}, for each $\epsilon > 0$, 
there is a time $t \in \RR$ so that $\anorm{e^{-itA(G)} - \gamma P_{\phi}} < \epsilon$,
for some $\gamma \in \TT$. Therefore, for each vertex $b$, we have
\begin{equation}
\anorm{e^{-itA(G)}\uket{b} - \gamma\uket{\phi(b)}}
	= \anorm{e^{-itA(G)}\uket{b} - \gamma\tilde{P}_{\phi}\uket{b}}
	\le \anorm{e^{-itA(G)} - \gamma P_{\phi}} \anorm{\uket{b}}
	< \epsilon,
\end{equation}
since $\anorm{\uket{b}}=1$. Thus, $G$ has pretty good state transfer between
vertices $b$ and $\phi(b)$.

Next, assume that $t \in \RR$ satisfies 
$\anorm{e^{-itA(G)} - \gamma\tilde{P}_{\phi}} < \epsilon/2$,
for some $\gamma \in \CC$.
For any $\tilde{P}_{\psi} \in \SwAut(G)$, we have
\begin{eqnarray}
\anorm{\tilde{P}_{\psi}\tilde{P}_{\phi} - \tilde{P}_{\phi}\tilde{P}_{\psi}}
	& \le & \anorm{\tilde{P}_{\psi}(\gamma\tilde{P}_{\phi}) - \tilde{P}_{\psi}e^{-itA(G)}}
		+ \anorm{e^{-itA(G)}\tilde{P}_{\psi} - \gamma\tilde{P}_{\phi}\tilde{P}_{\psi}} \\
	& \le & \anorm{\tilde{P}_{\psi}} \anorm{\gamma\tilde{P}_{\phi} - e^{-itA(G)}}
		+ \anorm{e^{-itA(G)} - \gamma\tilde{P}_{\phi}} \anorm{\tilde{P}_{\psi}} \\
	& < & \epsilon\anorm{\tilde{P}_{\psi}}.
\end{eqnarray}
Since $\anorm{\tilde{P}_{\psi}} \le 1$, we have
$\anorm{\tilde{P}_{\psi}\tilde{P}_{\phi} - \tilde{P}_{\phi}\tilde{P}_{\psi}} < \epsilon$,
for any $\epsilon > 0$. This implies that 
\begin{equation}
\anorm{\tilde{P}_{\psi}\tilde{P}_{\phi} - \tilde{P}_{\phi}\tilde{P}_{\psi}} = 0.
\end{equation}
Therefore, $\tilde{P}_{\psi}\tilde{P}_{\phi} = \tilde{P}_{\phi}\tilde{P}_{\psi}$.
Since $\tilde{P}_{\phi}$ commutes with each element of $\SwAut(G)$, it is in 
the center of $\SwAut(G)$.
\end{proof}

As a corollary to both Theorems \ref{thm:pgst-implies-permutation} and \ref{thm:pst-implies-permutation},
we show that if a graph $G$ has state transfer from vertex $a$ to vertex $b$ and there are two
automorphisms that map $a$ to $b$, then the two automorphisms are necessarily equal. This will
be useful later in showing that a non-trivial switching automorphism of a universal state transfer
graph has no fixed points.

\begin{corollary} \label{cor:pgst-kay}
Let $G$ be a graph with (perfect or pretty good) state transfer from vertex $a$ to vertex $\phi(a)$, 
for some switching automorphism $\phi \in \SwAut(G)$.
If $\psi \in \SwAut(G)$ is a switching automorphism for which $\psi(a) = \phi(a)$, then 
$\phi = \psi$.
\end{corollary}
\begin{proof}
Suppose $G$ has pretty good state transfer from $a$ to $\phi(a)$ and that $\psi(a)=\phi(a)$.
By Theorem \ref{thm:pgst-implies-permutation}, for any $\epsilon > 0$, 
there is a time $t \in \RR$ so that
\begin{equation}
\anorm{e^{-itA(G)} - \gamma_{1}\tilde{P}_{\phi}} \ < \ \frac{\epsilon}{2}
\ \ \
\mbox{ and }
\ \ \
\anorm{e^{-itA(G)} - \gamma_{2}\tilde{P}_{\psi}} \ < \ \frac{\epsilon}{2}.
\end{equation}
for some $\gamma_{1},\gamma_{2} \in \TT$.
Therefore, we have
\begin{equation}
\anorm{\gamma_{1}\tilde{P}_{\phi} - \gamma_{2}\tilde{P}_{\psi}}
\ \le \
\anorm{\gamma_{1}\tilde{P}_{\phi} - e^{-itA(G)}}
+
\anorm{e^{-itA(G)} - \gamma_{2}\tilde{P}_{\psi}}
\ < \
\epsilon.
\end{equation}
This holds for any $\epsilon > 0$ and thus $\tilde{P}_{\phi} = \gamma\tilde{P}_{\psi}$, 
for some $\gamma \in \CC$. This implies $\phi = \psi$.

If $G$ has perfect state transfer from $a$ to $\phi(a) = \psi(a)$ at time $t$,
then, by Theorem \ref{thm:pst-implies-permutation}, we have
$e^{-itA(G)} = \gamma_{1}\tilde{P}_{\phi} = \gamma_{2}\tilde{P}_{\psi}$.
This implies $\phi = \psi$.
\end{proof}

%%%%%%%%%%%%%%%%%%%%%%%%%%%%%%%%%%%%%%%%%%%%%%%%%%%%%%%%%%%%%%%%%%%%%%%%%%%%%%%%%%%%%%%%%%%%%%%%%%%%%%%%%%%%%%%%%%

\section{Universal Pretty Good State Transfer}

Here, we consider graphs with universal pretty good state transfer property.
We show strong spectral conditions on graphs with such universal property; namely, that 
each eigenvalue is simple and that each eigenvector has entries whose magnitudes are 
equal to each other. 
A $n \times n$ complex matrix $B$ {\em flat} if each entry of $B$ 
has the same complex magnitude. That is, $B$ is flat if
$|\mbraket{j}{B}{k}| = 1/\sqrt{n}$, for all $j,k$.

\begin{theorem} \label{thm:upgst-flat}
Let $G$ be a Hermitian graph with universal pretty good state transfer.
If $U$ is a unitary matrix which diagonalizes $A(G)$, then $U$ is flat.
\end{theorem}
\begin{proof}
Let $G$ be a $n$-vertex Hermitian graph with universal pretty good state transfer. 
Suppose $U$ is a unitary matrix whose $k$th column is given by $\ket{z_{k}}$,
where $U^{\dagger}A(G)U = \Lambda$, for some $\Lambda = \diag(\lambda_{1},\ldots,\lambda_{n})$.
Thus, $A(G) = \sum_{k=1}^{n} \lambda_{k}\ketbra{z_{k}}{z_{k}}$.
Since $G$ has universal pretty good state transfer, 
for each pair of vertices $a,b$, for any $\epsilon > 0$, there is a time $t \in \RR$ so that
\begin{eqnarray}
1-\epsilon
	& \le & |\tbracket{\buket{a}}{e^{-itA(G)}}{\buket{b}}| \\
	& = & \left|\sum_{k=1}^{n} e^{-it\lambda_{k}} 
		\bracket{\buket{a}}{\bvket{z_{k}}} \bracket{\bvket{z_{k}}}{\buket{b}}\right|,
		\ \ \mbox{ by spectral theorem } \\ 
	& \le & \sum_{k=1}^{n} |\bracket{\buket{a}}{\bvket{z_{k}}} \bracket{\buket{b}}{\bvket{z_{k}}}|,
		\ \ \mbox{ by triangle inequality } \\ 
	& \le & \sqrt{\sum_{k=1}^{n} |\bracket{\buket{a}}{\bvket{z_{k}}}|^{2}}
		\sqrt{\sum_{k=1}^{n} |\bracket{\buket{b}}{\bvket{z_{k}}}|^{2}},
		\ \ \mbox{ by Cauchy-Schwarz } \\ 
	& = & 1
\end{eqnarray}
since $|\bracket{\buket{a}}{\bvket{z_{k}}}|^{2} = 
\bracket{\buket{a}}{\bvket{z_{k}}}\bracket{\bvket{z_{k}}}{\buket{a}}$ 
and $\sum_{k=1}^{n} \ketbra{z_{k}}{z_{k}} = \II$.
Given that $\epsilon > 0$ can be made arbitrarily small, we have
\begin{equation}
\sum_{k=1}^{n} |\bracket{\buket{a}}{\bvket{z_{k}}}\bracket{\bvket{z_{k}}}{\buket{b}}| = 1.
\end{equation}
Since this implies that we achieve equality in Cauchy-Schwarz, 
we get $|\bracket{\buket{a}}{\bvket{z_{k}}}| = |\bracket{\buket{b}}{\bvket{z_{k}}}|$ for each $k$.
Since $a$ and $b$ are arbitrary, this shows that $U$ is flat.
\end{proof}

The proof method used in Theorem \ref{thm:upgst-flat} is originally due to Chris Godsil 
(who used a beautiful blend of triangle inequality (once) and Cauchy-Schwarz (twice) to 
analyze state transfer in quantum walk).

\begin{theorem} \label{thm:upgst-distinct-eigenvalues}
If $G$ is a Hermitian graph with universal pretty good state transfer, 
then $G$ has distinct eigenvalues.
\end{theorem}
\begin{proof}
Let $A(G)$ be the Hermitian adjacency matrix of a $n$-vertex graph $G$ 
and let $U$ be a unitary matrix that diagonalizes $A(G)$. 
Suppose the columns of $U$ are given by $\ket{z_{1}},\ldots,\ket{z_{n}}$ 
where $\ket{z_{k}}$ is the eigenvector corresponding to eigenvalue $\lambda_{k}$.
Thus, the $(j,k)$-entry of $U$ is given by 
$U_{j,k} = \bracket{\buket{j}}{\bvket{z_{k}}}$, for $j,k=1,\ldots,n$.

By Theorem \ref{thm:upgst-flat}, we know $U$ is flat.
Thus, $|\bracket{\buket{j}}{\bvket{z_{k}}}| = 1/\sqrt{n}$, for each $j,k$.
Without loss of generality, we assume $\bracket{\buket{1}}{\bvket{z_{k}}} = 1/\sqrt{n}$ 
for each $k$.
Otherwise, if $\bracket{\buket{1}}{\bvket{z_{k}}} = e^{i\theta_{k}}/\sqrt{n}$,
we may use the following unitary matrix 
$\tilde{U} = U\diag(\{e^{-i\theta_{k}}\})$ 
as our eigenbasis instead.
That is,
\begin{equation}
\tilde{U}
=
U
\begin{bmatrix}
e^{-i\theta_{1}} & 0 & 0 & \ldots & 0 \\
0 & e^{-i\theta_{2}} & 0 & \ldots & 0 \\
\vdots & \vdots & \vdots & \ldots & \vdots \\
0 & 0 & 0 & \ldots & e^{-i\theta_{n}}.
\end{bmatrix}
\end{equation}
This ensures that $\bracket{\buket{1}}{\bvket{z_{k}}} = 1/\sqrt{n}$, for each $k$.

Assume for contradiction that $\lambda_{1} = \lambda_{2}$.
Consider the following two alternative eigenvectors of $\lambda_{1}$ 
based on $\ket{z_{1}}$ and $\ket{z_{2}}$:
\begin{equation}
\ket{z_{+}} = \frac{1}{\sqrt{2}}(\ket{z_{1}} + \ket{z_{2}}),
\ \ \
\ket{z_{-}} = \frac{1}{\sqrt{2}}(\ket{z_{1}} - \ket{z_{2}}),
\end{equation}
Note that $\ket{z_{+}}$ and $\ket{z_{-}}$ are eigenvectors of $A(G)$ 
associated with eigenvalue $\lambda_{1}$.
Moreover, $\mathcal{B} = \{\ket{z_{+}},\ket{z_{-}}\} \cup \{\ket{z_{k}} : k=3,\ldots,n\}$ 
forms an orthonormal set of eigenvectors of $A(G)$. 
But the eigenbasis $\mathcal{B}$ is not flat, since 
\begin{equation}
\bracket{\buket{1}}{\bvket{z_{+}}} = \frac{1}{\sqrt{2}}
	\left[\bracket{\buket{1}}{\bvket{z_{1}}} + \bracket{\buket{1}}{\bvket{z_{2}}}\right]
	= \frac{\sqrt{2}}{\sqrt{n}}.
\end{equation}
This contradicts Theorem \ref{thm:upgst-flat} and proves the claim.
\end{proof}

\begin{corollary} \label{cor:pgst-commute-polynomial}
Let $G$ be a graph with universal pretty good state transfer.
Then, any matrix which commutes with $A(G)$ is a polynomial in $A(G)$.
\end{corollary}
\begin{proof}
Suppose $G$ is a graph with universal pretty good state transfer.
Assume $A(G) = \sum_{k=1}^{n} \lambda_{k} E_{k}$, where $E_{k}$ 
is a rank-one eigenprojector corresponding to the distinct eigenvalues 
$\lambda_{k}$ (by Theorem \ref{thm:upgst-distinct-eigenvalues}).
Moreover, each eigenprojector $E_{k}$ is a polynomial in $A(G)$ (see Godsil \cite{godsil93}).
So, for each $k$, let $E_{k} = p_{k}(A(G))$, for a polynomial $p_{k}(x) \in \RR[x]$ 

Let $B$ be a matrix which commutes with $A(G)$. 
Since $A(G)$ has distinct eigenvalues, each eigenspace of $A(G)$ is a $B$-invariant subspace. 
Thus, we have
\begin{equation}
B = \sum_{k=1}^{n} \mu_{k}E_{k} = \sum_{k=1}^{n} \mu_{k}p_{k}(A(G)),
\end{equation} 
where $\mu_{k}$ are the eigenvalues of $B$. 
This shows that $B$ is a polynomial in $A(G)$.
\end{proof}

Next, we consider the switching automorphism group of universal pretty good state transfer graphs. 
For such graphs, we show that the only switching automorphism that has fixed points is
the trivial automorphism. This result will be useful in proving stronger properties 
on the structure of the switching automorphism group of these graphs.

\begin{lemma} \label{lemma:upgst-no-fixed-point}
Let $G$ be a graph with universal pretty good state transfer.
Then, every nontrivial switching automorphism $\phi \in \SwAut(G)$ of $G$ has no fixed point.
\end{lemma}
\begin{proof}
Suppose that the switching automorphism $\phi$ has a fixed point at vertex $a$, 
that is, $\phi(a) = a$. 
Since $G$ has universal pretty good state transfer, $G$ has pretty good state transfer 
from $a$ to $a$.
By Corollary \ref{cor:pgst-kay}, since the trivial automorphism $\id$ satisfies 
$\id(a) = a$, we have $\phi = \id$.
\end{proof}

Using Lemma \ref{lemma:upgst-no-fixed-point}, we show that any two switching automorphisms 
of a graph with universal pretty good state transfer must commute with each other
and the number of those switching automorphisms divides the order of the underlying graph.

\begin{theorem} \label{thm:upgst-implies-abelian}
Let $G$ be a graph with universal pretty good state transfer.
Then, the switching automorphism group $\SwAut(G)$ is abelian and its order divides $|G|$.
\end{theorem}
\begin{proof}
For a switching automorphism $\tilde{P}_{\phi} \in \SwAut(G)$, let 
$\Fix(\phi) = \{a \in V : \phi(a) = a\}$ be the set of fixed points of $\phi$.
For a vertex $a$, let 
$\Orb(a) = \{\phi(a) : \phi \in \SwAut(G)\}$ be the orbit of $a$ under the action of $\SwAut(G)$. 
The set of orbits $\Orb(G) = \{\Orb(a) : a \in V\}$ forms a partition of $V$.
By Burnside's Lemma, we get 
\begin{equation}
|\Orb(G)| \cdot |\SwAut(G)| = \sum_{\phi} |\Fix(\phi)|.
\end{equation}
By Lemma \ref{lemma:upgst-no-fixed-point}, 
we have $|\Fix(\phi)| = n$, if $\phi = \id$, and $0$, otherwise;
thus, $\sum_{\phi} |\Fix(\phi)| = n$. 
This implies
\begin{equation}
|\Orb(G)| \cdot |\SwAut(G)| = n.
\end{equation}
This shows $|\SwAut(G)|$ divides $n$ since the number of orbits is an integer.

Next, we show that $\SwAut(G)$ is abelian.
Consider two switching automorphisms $\tilde{P}_{\phi},\tilde{P}_{\psi} \in \SwAut(G)$. 
By definition, both commute with $A(G)$, and by Corollary \ref{cor:pgst-commute-polynomial},
both are polynomials in $A(G)$. 
This shows $\tilde{P}_{\phi}$ and $\tilde{P}_{\psi}$ commute.
\end{proof}

%%%%%%%%%%%%%%%%%%%%%%%%%%%%%%%%%%%%%%%%%%%%%%%%%%%%%%%%%%%%%%%%%%%%%%%%%%%%%%%%%%%%%%%%%%%%%%%%%%%%%%%%%%%%%%%%%%

\section{Universal Perfect State Transfer}

In this section, we describe some properties of graphs which have universal 
{\em perfect} state transfer. Due to the more stringent requirement on state transfer, 
we are able to prove stronger properties on the structure of these graphs. 
First, we show that universal perfect state transfer graphs are periodic 
(in the strong sense).

\begin{lemma} \label{lemma:upst-periodic}
Let $G$ be a graph with universal perfect state transfer.
Then, $G$ is periodic.
\end{lemma}
\begin{proof}
Let $A(G)$ be the Hermitian adjacency matrix of $G$.
Since $G$ has universal perfect state transfer, for each vertex $a$ of $G$,
there is a time $t$ so that $e^{-itA(G)}\uket{a} = \gamma\uket{a}$,
for some $\gamma \in \CC$ with $|\gamma|=1$.
By Theorem \ref{thm:pst-implies-permutation}, using the trivial switching automorphism
$\id(a) = a$, we have
\begin{equation}
e^{-itA(G)} = \alpha\II,
\end{equation}
since $\tilde{P}_{\id} = \alpha\II$ for some $\alpha \in \TT$.
This shows that $G$ is periodic.
\end{proof}

Using a result of Godsil \cite{godsil-ejc11}, we note that the ratio of any
two eigenvalues of a universal perfect state transfer graph must be rational,
provided the denominator is not a zero eigenvalue.

\begin{corollary} \label{cor:rational-eigenvalue-ratio}
Let $G$ be a graph with universal perfect state transfer
whose adjacency matrix $A(G)$ satisfies $\Tr A(G) = 0$.
Then, for any eigenvalues $\lambda_{j} \neq \lambda_{k}$, with $\lambda_{k} \neq 0$,
the ratio $\lambda_{j}/\lambda_{k}$ is rational.
\end{corollary}
\begin{proof}
This is an immediate corollary of Theorem 3.1 in Godsil \cite{godsil-ejc11}.
\end{proof}

In the next theorem, we show that a universal perfect state transfer graph 
has a cyclic switching automorphism group.
But, first we prove a useful lemma which shows that the set of times when
the quantum walk ``visits'' the switching automorphism group of a universal
perfect state transfer graph is not dense.

\begin{lemma} \label{lemma:upst-not-dense}
Let $G$ be a graph with universal perfect state transfer.
Then, the set
\begin{equation}
\Gamma = \left\{t \in \RR : (\exists \phi \in \SwAut(G))(\exists \gamma \in \TT)
	\left[ e^{-itA(G)} = \gamma\tilde{P}_{\phi} \right] \right\}
\end{equation}
is a discrete additive subgroup of $\RR$.
\end{lemma}
\begin{proof}
Let $A$ be the adjacency matrix of $G$ and let $U(t) = e^{-itA}$.
Since $\id \in \SwAut(G)$, the set
$T_{\id} = \{t \in \RR : (\exists\gamma\in\TT) \ U(t) = \gamma\II\}$
is a subset of $\Gamma$. 
Godsil \cite{godsil-dm11} showed that $T_{\id}$ is a discrete additive subgroup 
of $\RR$. It is additive since $U(t_{1}+t_{2}) = U(t_{1})U(t_{2})$, for any
times $t_{1},t_{2}$.
%%%%%%%%%%%%%%%%%%%%%%%%%%%%%%%%%%%%%%%%%%%%%%%%%%%%%%%%%%%%%%%%%%%%%%%%%%%%%%
%%% Lie-theoretic argument %%%
%%%%%%%%%%%%%%%%%%%%%%%%%%%%%%%%%%%%%%%%%%%%%%%%%%%%%%%%%%%%%%%%%%%%%%%%%%%%%%
The argument for why it is discrete is as follows.
If $T_{\id}$ is dense, then it contains a sequence $\{t_{k}\}_{k=0}^{\infty}$ 
where $t_{k} \rightarrow 0$ as $k \rightarrow \infty$.
Note that $\lim_{k \rightarrow \infty} (U(t_{k})-\II)/t_{k}$ exists and 
it equals $U'(0) = -iA$, since $U(t)$ is differentiable.
If $U(t_{k}) = \gamma_{k}\II$, for $\gamma_{k} \in \TT$,
the above limit is also equal to
$\lim_{k \rightarrow \infty} t_{k}^{-1}(\gamma_{k}-1)\II$. 
But, $\anorm{t_{k}^{-1}(\gamma_{k}-1)\II - (-iA)}$ is bounded away from $0$
if $A$ contains nonzero off-diagonal entries.
%%%%%%%%%%%%%%%%%%%%%%%%%%%%%%%%%%%%%%%%%%%%%%%%%%%%%%%%%%%%%%%%%%%%%%%%%%%%%%

We now show $T_{\phi}$ is discrete, for each $\phi \in \SwAut(G)$.
Since $\SwAut(G)$ is finite, there is a positive integer $m$ so that 
$\tilde{P}_{\phi}^{m} = \alpha\II$, for some $\alpha \in \TT$.
Suppose, for contradiction, that $T_{\phi}$ is dense.
Then, for any $\delta > 0$, there is a time $t \in (0,\delta)$ so that
$U(t) = \beta\tilde{P}_{\phi}$, for some $\beta \in \TT$.
This shows that $U(mt) = (\alpha\beta^{m})\II$, which implies that $T_{\id}$ is
also dense. This contradiction shows that $T_{\phi}$ is discrete.
\end{proof}

\medskip

\begin{theorem} \label{thm:upst-implies-cyclic}
Let $G$ be a graph with universal perfect state transfer.
Then, the switching automorphism group $\SwAut(G)$ is cyclic.
\end{theorem}
\begin{proof}
Let $G$ be a graph with universal perfect state transfer 
whose Hermitian adjacency matrix is $A(G)$.
In the proof, we use $\gamma_{0},\gamma_{1},\gamma_{2},\ldots$
to represent complex numbers of unit modulus.
We consider the set $\Gamma$ of times where the quantum walk ``visits''
the switching automorphism group:
\begin{equation}
\Gamma = 
\left\{t \in \RR^{+} : (\exists\phi \in \SwAut(G))(\exists \gamma \in \TT)
	\left[e^{-itA(G)} = \gamma\tilde{P}_{\phi}\right] \right\}.
\end{equation}
By Lemma \ref{lemma:upst-not-dense}, $\Gamma$ is a discrete subgroup of $\RR$
and thus it has a least positive element. Let 
\begin{equation}
\tstar = \min\Gamma 
\end{equation}
and let $\vstar \in \SwAut(G)$ be the switching automorphism 
for which $e^{-i\tstar A(G)} = \gamma_{0}\tilde{P}_{\vstar}$.
We show that $\tilde{P}_{\vstar}$ generates $\SwAut(G)$.

Note that if $e^{-it_{1} A(G)} = \gamma_{1}\tilde{P}_{\psi}$, 
for some $\tilde{P}_{\psi} \in \SwAut(G)$,
then $t_{1}$ is necessarily an integer multiple of $\tstar$. 
Otherwise, if $m\tstar < t_{1} < (m+1)\tstar$, 
for some nonnegative integer $m$, then
\begin{eqnarray}
e^{-i(t_{1} - m\tstar)A(G)}
	& = & e^{-it_{1} A(G)} [e^{i\tstar A(G)}]^{m} \\
	& = & \gamma_{1}\tilde{P}_{\psi}(\gamma_{0}\tilde{P}_{\vstar})^{m} \\
	& = & (\gamma_{0}^{m}\gamma_{1}) \ \tilde{P}_{\psi}\tilde{P}_{(\vstar)^m} \\
	& = & (\gamma_{0}^{m}\gamma_{1}\gamma_{2}) \ \tilde{P}_{\psi (\vstar)^{m}}.
\end{eqnarray}
This contradicts the choice of $\tstar$ since $0 < t_{1} - m\tstar < \tstar$.

We show that for each $\tilde{P}_{\phi} \in \SwAut(G)$, there is a time $t$ 
for which $e^{-it A(G)} = \gamma\tilde{P}_{\phi}$, for some $\gamma \in \TT$.
Fix a vertex $a$ of $G$ and let $b = \phi(a)$ with 
$\tilde{P}_{\phi}\uket{a} = \gamma_{3}\uket{b}$.
Since $G$ has universal perfect state transfer, there is a time $\tilde{t}$ where
\begin{equation}
e^{-i\tilde{t}A(G)}\uket{a} = \gamma_{4}\uket{b} = (\gamma_{3}^{-1}\gamma_{4}) \tilde{P}_{\phi}\uket{a}.
\end{equation}
By Theorem \ref{thm:pst-implies-permutation}, we have
$e^{-i\tilde{t}A(G)} = \gamma\tilde{P}_{\phi}$, for some $\gamma \in \TT$.

Thus, for each $\tilde{P}_{\phi} \in \SwAut(G)$, there is a time $\tilde{t}$,
where $e^{-i\tilde{t}A(G)} = \gamma\tilde{P}_{\phi}$ and $\tilde{t} = m\tstar$, 
for some integer $m \in \ZZ$ and $\gamma \in \TT$. Therefore,
\begin{equation}
\gamma\tilde{P}_{\phi} 
	= e^{-i\tilde{t}A(G)} 
	= (e^{-i\tstar A(G)})^{m} 
	= \gamma_{0}^{m} \tilde{P}_{\vstar}^{m}.
\end{equation}
This implies $\phi = (\vstar)^{m}$ which shows that $\SwAut(G) = \avg{\vstar}$.
\end{proof}

By Theorem \ref{thm:upgst-implies-abelian}, the order of the switching automorphism group 
of a universal perfect state transfer graph must divide the order of the graph.
The next theorem characterizes those graphs whose switching automorphism group
has order equal to the order of the graph.

\begin{theorem}
Let $G$ be a $n$-vertex graph with universal perfect state transfer.
Then, $\SwAut(G)$ is cyclic of order $n$ 
if and only if 
$G$ is switching isomorphic to a circulant.
\end{theorem}
\begin{proof}
Assume that $G$ is a graph with universal perfect state transfer
whose adjacency matrix is the Hermitian matrix $A(G)$.

($\Longrightarrow$)
Suppose that the switching automorphism group $\SwAut(G)$ is cyclic of order $n$.
Let $\vstar$ be a switching automorphism that generates $\SwAut(G)$, that is, 
$\SwAut(G) = \avg{\vstar}$.
Since $\tilde{P}_{\vstar}$ has no fixed points (by Lemma \ref{lemma:upgst-no-fixed-point})
and has order $n$, its cycle structure is of a $n$-cycle. 
Without loss of generality, by reordering, we may assume that $\vstar = (1 \ 2 \ \ldots \ n)$.
Since $\tilde{\Theta}_{n} = \tilde{P}_{\vstar} \in \SwAut(G)$, 
it commutes with $A(G)$, and hence they share the same circulant eigenbasis. 
This shows that $A(G)$ is a circulant up to switching equivalence.

($\Longleftarrow$)
Suppose that $G$ is switching isomorphic to a circulant.
Thus, $A(G) = \tilde{P}^{\dagger}C\tilde{P}$ for some circulant matrix $C$ and monomial matrix $\tilde{P}$.
By Theorem \ref{thm:upst-implies-cyclic}, we know $\SwAut(G)$ is cyclic whose order divides $n$.
Since any circulant matrix is a polynomial in $\Theta_{n}$, 
suppose $C = f(\Theta_{n})$ for some polynomial $f(x) \in \CC[x]$.
Consider the monomial matrix $\tilde{Q} = \tilde{P}^{\dagger}\Theta_{n}\tilde{P}$.
Note that $A(G)$ commutes with $\tilde{Q}$ since
$A(G)$ is a polynomial in $\tilde{Q}$:
\begin{equation}
A(G) 
	= \tilde{P}^{\dagger}f(\Theta_{n})\tilde{P}
	= f(\tilde{P}^{\dagger}\Theta_{n}\tilde{P})
	= f(\tilde{Q}).
\end{equation}
This implies that $\tilde{Q}$ is a switching automorphism of $G$.
Note $\tilde{Q}$ is of order $n$, since $\tilde{\Theta}_{n}$ is.
So, the order of $\SwAut(G)$ is $n$ since it has an element of order $n$. 
\end{proof}

%%%%%%%%%%%%%%%%%%%%%%%%%%%%%%%%%%%%%%%%%%%%%%%%%%%%%%%%%%%%%%%%%%%%%%%%%%%%%%%%%%%%%%%%%%%%%%%%%%%%%%%%%%%%%%%%%%

\section{Explicit Constructions}

Our goal in this section is to show infinite families of Hermitian graphs with universal 
pretty good state transfer. Given our earlier characterizations, it is not surprising that
these families are based on complex circulants.

Let $\SC_{n}$ be a $n$-vertex Hermitian graph whose adjacency matrix is the following circulant matrix:
\begin{equation}
A(\SC_{n}) = 
\begin{bmatrix}
0 & -i & 0 & \ldots & 0 & i \\
i & 0 & -i & \ldots & 0 & 0 \\
0 & i & 0 & \ldots & 0 & 0 \\
\vdots & \vdots & \vdots & \vdots & \vdots & \vdots \\
0 & 0 & 0 & \ldots & 0 & -i \\ 
-i & 0 & 0 & \ldots & i & 0
\end{bmatrix}
\end{equation}
We note that $A(\SC_{n}) = i\Theta_{n} - i\Theta_{n}^{T}$,
where $\Theta_{n} = \Circ(0,\ldots,0,1)$ as defined in Section \ref{section:preliminaries}.

\begin{figure}[t]
\begin{center}
\begin{tikzpicture}[->, >=stealth', shorten >=2pt, scale=1.5][h!!!!!!!!!]
\foreach \c [count=\x from 1] in {0, ..., 4} %list for label, count for degree
{
	%\node at (\x*72+18:1)[circle, circular drop shadow, fill={blue!20}, draw][scale=0.75]{$\c$};	
	\draw[->, >=latex] ({(\x-1)*72+18+9}:1)
	arc ({(\x-1)*72+18+9}:{(\x*72+18-9}:1); %arc
	\draw [->, >=latex]({\x*72+18-9}:1) -> ({(\x-1)*72+18+8}:1); %straight edge
	\node at (\x*72+50:0.65)[circle][scale=0.75]{$-i$}; % edge label
	\node at (\x*72+50:1.25)[circle][scale=0.75]{$i$};	
}	
\foreach \c [count=\x from 1] in {0, ..., 4} %list for label, count for degree
{
	\node at (\x*72+18:1)[circle, circular drop shadow, fill={blue!20}, draw][scale=0.75]{$\c$};	
}	
\end{tikzpicture}
\caption{The Hermitian $\SC_{p}$ has universal pretty good state transfer, for each prime $p$.}
\end{center}
\end{figure}
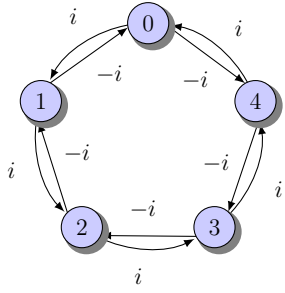

We say that a set of real numbers $\lambda_{1},\ldots,\lambda_{m}$ is
{\em linearly independent over the rationals} if whenever
$\sum_{k=1}^{m} a_{k}\lambda_{k} = 0$, for rational numbers $a_{1},\ldots,a_{m} \in \QQ$, 
then $a_{k} = 0$, for each $k=1,\ldots,m$.
As in Godsil \etal \cite{gkss12}, we apply Kronecker's theorem from number theory. 
The particular version of Kronecker's theorem useful for our case is stated 
in the following (see Theorem 444 in Hardy and Wright \cite{hardy-wright}). 

\begin{theorem} \label{thm:kronecker}
(Kronecker) 
Let $\lambda_{1},\ldots,\lambda_{m}$ be linearly independent over $\QQ$.
Let $\alpha_{1},\ldots,\alpha_{m}$ be arbitrary real numbers, 
and let $T$ and $\epsilon$ be positive reals.
Then, there is a real number $t \in \RR$ and integers $p_{1},\ldots,p_{m} \in \ZZ$ 
so that $t > T$ and
\begin{equation}
|t\lambda_{k} - \alpha_{k} - p_{k}| < \epsilon,
\end{equation}
for each $k=1,\ldots,m$.
\end{theorem}

The following proposition is a crucial ingredient of our main result in this section.

\begin{proposition} \label{prop:spectra-hermitian-cycle}
For any prime $p > 2$, the set of real numbers
\begin{equation}
\left\{ \sin\left(\frac{2\pi k}{p}\right) : k=1,\ldots,(p-1)/2\right\}
\end{equation}
is linearly independent over the rationals.
\end{proposition}
\begin{proof}
For each integer $k$, we may express $e^{ik\theta}$ in two different ways:
\begin{equation}
e^{ik\theta} 
= 
(\cos\theta + i\sin\theta)^{k}
=
\cos(k\theta) + i\sin(k\theta).
\end{equation}
Thus, we have $\sin(k\theta) = \Imag(e^{ik\theta})$.
Therefore,
\begin{eqnarray}
\sin(k\theta) 
	& = & \sum_{j=0}^{k} \binom{k}{j} \iverson{j \mbox{\footnotesize ~odd}} 
			(-1)^{(j-1)/2} (\cos\theta)^{k-j} (\sin\theta)^{j} \\
	& = & \sum_{\ell=0}^{\lfloor \frac{k-1}{2}\rfloor} \binom{k}{2\ell+1} 
			(-1)^{\ell} (\cos\theta)^{k-(2\ell+1)} (\sin\theta)^{2\ell+1}.
\end{eqnarray}
For our purposes, we let $\theta = 2\pi/p$, where $p$ is an odd prime.
For $k$ odd, we may replace $\cos\theta$ in the above equation 
with $\sqrt{1-\sin^{2}\theta}$, and it shows that 
$\sin(2\pi k/p)$ is a polynomial in $\sin(2\pi/p)$ of degree at most $k$.
For $k$ even, we note that
\begin{equation}
\sin\left(\frac{2\pi k}{p}\right)
=
- \sin\left(\frac{2\pi (p-k)}{p}\right),
\end{equation}
and now $p-k$ is odd.
In either case, we have that $\sin(2\pi k/p)$ is a polynomial in
$\sin(2\pi/p)$ of degree at most $\max\{k,p-k\}$.
Note that $\max\{k,p-k\} \le p-2$ for each $k=1,\ldots,(p-1)/2$.

Suppose that the rational numbers $a_{1},\ldots,a_{(p-1)/2}$ satisfy
\begin{equation}
\sum_{k=1}^{(p-1)/2} a_{k}\sin\left(\frac{2\pi k}{p}\right) = 0.
\end{equation}
Since $\sin(2\pi k/p)$ may be written as $f_{k}(\sin(2\pi/p))$ where
$f_{k}(x)$ is a polynomial of degree $d_{k}$, for $d_{k} \le p-2$, 
if $a_{k} \neq 0$ for some collection of indices $k$, then
\begin{equation}
	\sum_{k: a_{k} \neq 0} a_{k} f_{k}\left(\sin\left(\frac{2\pi}{p}\right)\right) = 0.
\end{equation}
Note that the left-hand side is a polynomial of degree at most $p-2$
that vanishes at $\sin(2\pi/p)$.
But, it is known that the minimal polynomial of $\sin(2\pi/p)$ 
is of degree $p-1$ (see \cite{bd04}).
Therefore, we must have $a_{k} = 0$ for all $k=1,\ldots,(p-1)/2$.
\end{proof}

We are now ready to prove our main result about the family of graphs $\SC_{p}$.

\begin{theorem} \label{thm:circulant-with-upgst}
For each prime $p$, the graph $\SC_{p}$ has universal pretty good state transfer.
\end{theorem}
\begin{proof}
Since $\SC_{p}$ is a circulant, its eigenvalues
are given by (see Equation (\ref{eqn:circulant-eigenvalue})):
\begin{equation}
\lambda_{k} = -i\omega_{p}^{k} + i\omega_{p}^{-k} = 2\sin\left(\frac{2\pi k}{p}\right),
\end{equation}
for $k=0,\ldots,p-1$, where $\omega_{p} = \exp(2\pi i/p)$.
Let $j$ be an arbitrary vertex of $\SC_{p}$. Then,
\begin{equation}
\tbracket{\buket{j}}{e^{-itA(\SC_{p})}}{\buket{0}}
= 
\sum_{k=0}^{p-1} e^{-it\lambda_{k}} \bracket{\buket{j}}{\bvket{F_{k}}}	
	\bracket{\ket{F_{k}}}{\uket{0}}
=
\frac{1}{p} \sum_{k=0}^{p-1} e^{-it\lambda_{k}} \omega_{p}^{jk},
\end{equation}
where $\ket{F_{k}}$ is the $k$th circulant eigenvector 
which satisfies $\bracket{\buket{j}}{\bvket{F_{k}}} = \omega_{p}^{jk}/\sqrt{p}$.
By using $\lambda_{p-k} = -\lambda_{k}$, which holds for $k=1,\ldots,(p-1)/2$,
we get
\begin{eqnarray} \label{eqn:upgst-prime-cycle}
\tbracket{\buket{j}}{e^{-itA(\SC_{p})}}{\buket{0}}
	& = & \frac{1}{p}
		\left[ 1 + \sum_{k=1}^{(p-1)/2} 
					\left(
					e^{-it\lambda_{k}} \omega_{p}^{jk}
					+
					e^{it\lambda_{k}} \omega_{p}^{-jk}
					\right)
		\right] \\
	& = & \frac{1}{p}
		\left[ 1 + \sum_{k=1}^{(p-1)/2} 
					2\cos\left(t\lambda_{k} - \frac{2\pi jk}{p}\right)
		\right]
\end{eqnarray}
By Proposition \ref{prop:spectra-hermitian-cycle}, the eigenvalues $\lambda_{k} = \sin(2\pi k/p)$,
for $k=1,\ldots,(p-1)/2$, are linearly independent over the rationals;
thus, the numbers $\lambda_{k}/2\pi$ are also linearly independent over $\QQ$.
Therefore, by Theorem \ref{thm:kronecker} (Kronecker's Theorem), 
for any positive reals $T$ and $\epsilon$, there is a time $t > T$ and integers $m_{k} \in \ZZ$ so that
\begin{equation}
\left|t\left(\frac{\lambda_{k}}{2\pi}\right) - \frac{jk}{p} - m_{k}\right| < \frac{\epsilon}{2\pi}
\ \
\Longleftrightarrow
\ \
\left|t\lambda_{k} - \frac{2\pi jk}{p} - 2\pi m_{k}\right| < \epsilon,
\end{equation}
for all $k=1,\ldots,(p-1)/2$.
We need to prove that for any $\delta > 0$, there is a time $t$ 
for which $|\tbracket{\buket{j}}{e^{-itA(\SC_{p})}}{\buket{0}}| \ge 1-\delta$. 

So, let $x_{k} = t\lambda_{k} - 2\pi jk/p - 2\pi m_{k}$, for $k=1,\ldots,(p-1)/2$, 
where we will view $x_{k}$ as a function of $t$.
Since $\cos(x)$ is continuous, for any given $\delta > 0$, there is a $\epsilon > 0$  
so that $\cos(x_{k}) > 1 - \delta$, whenever $|x_{k}| < \epsilon$.
Here, Kronecker's Theorem allows us to choose a time $t$ so that $|x_{k}| < \epsilon$.
Applying this to Equation (\ref{eqn:upgst-prime-cycle}), we get
\begin{equation}
\tbracket{\buket{j}}{e^{-itA(\SC_{p})}}{\buket{0}}
\ \ge \
\frac{1}{p}\left(1 + (p-1)(1-\delta)\right)
\ \ge \
1-\delta.
\end{equation}
This shows that $\SC_{p}$ has universal pretty good state transfer.
\end{proof}

We observe in the following corollary that arguments used to prove
Theorem \ref{thm:circulant-with-upgst} also apply to a Cartesian
bunkbed of $\SC_{p}$, for primes $p \ge 5$.
Recall that given two graphs $G_{1}$ and $G_{2}$, their Cartesian product 
$G_{1} \cart G_{2}$ is defined as the graph whose adjacency matrix is 
$A(G_{1}) \otimes \II + \II \otimes A(G_{2})$.

\begin{corollary}
For each prime $p \ge 5$, the graph $K_{2} \cart \SC_{p}$ has universal pretty good state transfer.
\end{corollary}
\begin{proof}
For a prime $p \ge 5$, let $\GG$ denote the graph $K_{2} \cart \SC_{p}$
whose adjacency matrix is
\begin{equation}
A(\GG) = A(K_{2}) \otimes \II_{p} + \II_{2} \otimes A(\SC_{p}).
\end{equation}
Given that this is a sum of commuting matrices, we have
\begin{equation}
e^{-itA(\GG)} 
	= e^{-it(A(K_{2}) \otimes \II_{p})} e^{-it(\II_{2} \otimes A(\SC_{p}))}
	= e^{-itA(K_{2})} \otimes e^{-itA(\SC_{p})}.
\end{equation}
By the proof of Theorem \ref{thm:circulant-with-upgst},
the eigenvalues of $\SC_{p}$ are defined by the values $\sin(2\pi k/p)$,
for $k=1,\ldots,(p-1)/2$.
Since $\sin(2\pi/n)$ is rational if and only if $n=1,2,3,4$, or $6$ (see \cite{bd04}), 
the set 
\begin{equation}
\Lambda = \{1\} \cup \{\sin(2\pi k/p) : 1 \le k \le (p-1)/2)\}
\end{equation}
is linearly independent over the rationals.

The quantum walk on $\GG$ from $(0,0)$ to $(b,j)$, for $b \in \zo$ and $j \in \ZZ_{p}$,
is given by
\begin{eqnarray} \label{eqn:cartesian-split}
\tbracket{\buket{b,j}}{e^{-itA(\GG)}}{\buket{0,0}}
	& = & \tbracket{\buket{b,j}}{(e^{-it A(K_{2})} \otimes e^{-it A(\SC_{p})})}{\buket{0,0}} \\
	& = & \tbracket{\buket{b}}{e^{-itA(K_{2})}}{\buket{0}} 
		\cdot \tbracket{\buket{j}}{e^{-itA(\SC_{p})}}{\buket{0}}.
\end{eqnarray}
On the other hand, a quantum walk on $K_{2}$ is given by
\begin{equation}
\tbracket{\buket{b}}{e^{-itA(K_{2})}}{\buket{0}}
	= \ubra{b}
		\begin{bmatrix}
		\cos(t) & -i\sin(t) \\
		-i\sin(t) & \cos(t)
		\end{bmatrix}
		\uket{0}
	= \left\{\begin{array}{ll}
		\cos(t) & \mbox{ if $b=0$ } \\
		-i\sin(t) & \mbox{ if $b=1$ }
		\end{array}\right.
\end{equation}
Thus, if $t \approx (2\ZZ + b)\pi/2$, then 
$|\tbracket{\buket{b}}{e^{-itA(K_{2})}}{\buket{0}}| \approx 1$.

By Kronecker's Theorem, since the set $\Lambda$ is linearly independent over $\QQ$, 
for any $\epsilon > 0$, there is a time $t$ so that 
the following inequalities are simultaneously satisfied:
\begin{equation} \label{eqn:approx-k2}
\left|t - \frac{(2+b)\pi}{2} - 2\pi m_{0}\right| \ < \ \epsilon
\end{equation}
\begin{equation} \label{eqn:approx-cp}
\left|t\lambda_{k} - \frac{2\pi jk}{p} - 2\pi m_{k}\right| \ < \ \epsilon
\end{equation}
where $m_{k}$ are integers, $k=0,1,\ldots,(p-1)/2$.
Here, $\epsilon > 0$ was chosen so that:
\begin{enumerate}
\item[(a)] $(\cos(t)\iverson{b=0} + \sin(t)\iverson{b=1}) > 1 - \delta$; and
\item[(b)] $\cos(x_{k}) > 1 - \delta$, where $x_{k} = t\lambda_{k} - 2\pi jk/p - 2\pi m_{k}$, 
	for $k=1,\ldots,(p-1)/2$.
\end{enumerate}
By the proof of Theorem \ref{thm:circulant-with-upgst}, 
condition (b) above ensures that 
$|\tbracket{\buket{j}}{e^{-itA(\SC_{p})}}{\buket{0}}| > 1-\delta$.
Also, condition (a) guarantees that 
$|\tbracket{\buket{b}}{e^{-itA(K_{2})}}{\buket{0}}| > 1-\delta$.
By Equation (\ref{eqn:cartesian-split}), we have
\begin{equation}
|\tbracket{\buket{b,j}}{e^{-itA(\GG)}}{\buket{0,0}}|
	\ge (1-\delta)^{2} 
	\ge 1-2\delta.
\end{equation}
This proves the claim.
\end{proof}

%%%%%%%%%%%%%%%%%%%%%%%%%%%%%%%%%%%%%%%%%%%%%%%%%%%%%%%%%%%%%%%%%%%%%%%%%%%%%%%%%%%%%%%%%%%%%%%%%%%%%%%%%%%%%%%%%%

\section{Generalized Circulants}

We study universal state transfer on graphs that are switching equivalent to circulants. 
For this class of graphs, we provide a spectral characterization for universal perfect state transfer. 
But, first we prove a useful lemma that reduces universal perfect state transfer to perfect state transfer 
between certain pairs of vertices.

\begin{lemma} \label{lemma:upst-zero-to-one}
Let $G$ be a $n$-vertex graph that is switching equivalent to a circulant.
Then, 
$G$ has universal perfect state transfer 
if and only if 
$G$ has perfect state transfer from $0$ to $j$,
for some $j \in \{1,\ldots,n-1\}$ which satisfies $\gcd(j,n)=1$.
\end{lemma}
\begin{proof}
Suppose $G$ is a graph on $n$ vertices which satisfies $A(G) = D^{\dagger}CD$,
where $D$ is a diagonal matrix with unit modulus complex entries and $C$ is a circulant matrix.
We only need to prove the sufficient condition since the necessary condition follows by definition.
Assume $G$ has perfect state transfer from $0$ to $j$, where $\gcd(j,n)=1$; that is,
\begin{equation}
e^{-itA(G)}\uket{0} = \gamma\uket{j},
\end{equation}
for some $\gamma \in \TT$.
By Theorem \ref{thm:pst-implies-permutation}, this implies that
\begin{equation}
e^{-itA(G)} = \gamma\tilde{P}_{\sigma},
\end{equation}
where $\sigma$ is a permutation defined by $\sigma(x) \equiv x+j\pmod{n}$.
Thus, $\sigma = (0 \ j \ 2j \ \ldots \ (n-1)j)$.
Since $\gcd(j,n)=1$, there is an integer $m$ so that $mj \equiv 1\pmod{n}$.
This implies that $\sigma^{m}(0) = 1$ and $\sigma^{mk}(0) = k$, for any $k$. 
To show that $G$ has universal perfect state transfer, it suffices to show 
that $G$ has perfect state transfer from $0$ to every other vertex, say $k$.
We have
\begin{equation}
(e^{-itA(G)})^{mk}\uket{0} 
	= (\gamma\tilde{P}_{\sigma})^{mk}\uket{0}
	= \gamma' P_{\sigma}^{mk}\uket{0}
	= \gamma'\uket{k}
\end{equation}
for some $\gamma' \in \TT$.
\end{proof}

We are now ready to prove the spectral characterization of universal perfect
state transfer graphs that are switching equivalent to circulants.

\begin{theorem}
Let $G$ be a $n$-vertex graph that is switching equivalent to a circulant.
Then, 
$G$ has universal perfect state transfer 
if and only if 
each eigenvalue of $G$ is of the form
\begin{equation} \label{eqn:circulant-eigenvalue-form}
\lambda_{k} = \alpha + \beta(jk + c_{k}n),
\end{equation}
where $\alpha, \beta \in \RR$, with $\beta > 0$, $j \in \ZZ$ is relatively prime to $n$,
and $c_{k} \in \ZZ$.
\end{theorem}
\begin{proof}
Let $G$ be a graph on $n$ vertices that is switching equivalent to a circulant.
Thus, $G$ has the circulant eigenbasis $\{\ket{F_{k}} : k=0,\ldots,n-1\}$, 
where $\vbbravec{j}{F_{k}} = e^{2\pi ijk/n}/\sqrt{n}$. 

We will make two simplifying assumptions on the eigenvalues of $G$ without loss of generality.
First, note if $\lambda_{k}$ is an eigenvalue of $A(G)$,
then $\lambda_{k} + \alpha$ is an eigenvalue of $A(G) + \alpha\II$.
This allows us to assume $\alpha = 0$ in Equation (\ref{eqn:circulant-eigenvalue-form}), 
since this diagonal shift merely introduces an irrelevant phase factor in a quantum walk:
\begin{equation}
e^{-it(A(G) - \alpha\II)} = e^{i\alpha}e^{-itA(G)}.
\end{equation}
Due to this invariance under diagonal shifts, we will also assume $\lambda_{0}=0$.
Second, note if $\lambda_{k}$ is an eigenvalue of $A(G)$,
then $\beta\lambda_{k}$ is an eigenvalue of $\beta A(G)$.
This allows us to assume $\beta = 1$ in Equation (\ref{eqn:circulant-eigenvalue-form}), since 
\begin{equation}
e^{-it\beta A(G)} = e^{-i(t\beta)A(G)},
\end{equation}
which shows that the time-scaling factor $\beta$ can be absorbed into our time parameter.

Using the above assumptions, the eigenvalues of $G$ are of the form
\begin{equation}
\lambda_{k} = jk + c_{k}n,
\end{equation}
where $c_{k} \in \ZZ$ are integers and $j \in \ZZ$ is an integer with $\gcd(j,n)=1$.

\bigskip

($\Longrightarrow$) Suppose $G$ has universal perfect state transfer.
Thus, for any $j$ with $\gcd(j,n)=1$, the graph $G$ has perfect state transfer 
from $0$ to $j$ at some time $t_{j}$. This implies
\begin{equation}
|\tbracket{\buket{j}}{e^{-it_{j}A(G)}}{\buket{0}}|
	= \left| \sum_{k=0}^{n-1} e^{-it_{j}\lambda_{k}} 
		\bracket{\buket{j}}{\bvket{F_{k}}} \bracket{\bvket{F_{k}}}{\buket{0}} \right|
	= \left| \frac{1}{n} \sum_{k=0}^{n-1} e^{-i(t_{j}\lambda_{k} - 2\pi jk/n)} \right|
	= 1.
\end{equation}
By Lemma \ref{lemma:complex-sum-forcing}, there is a real number $\theta \in \RR$
so that for each $k$, we have
\begin{equation} 
t_{j}\lambda_{k} - \frac{2\pi jk}{n} \equiv \theta \pmod{2\pi}.
\end{equation}
Due to invariance under time-scaling, we assume without loss of generality that $t_{j} = 2\pi/n$.
Therefore,
\begin{equation}
\frac{2\pi}{n}(\lambda_{k} - jk) \equiv \frac{2\pi}{n}(\lambda_{\ell} - j\ell) \pmod{2\pi}.
\end{equation}
But since $\lambda_{0} = 0$, we get
\begin{equation}
\frac{2\pi}{n}(\lambda_{k} - jk) \equiv 0 \pmod{2\pi}
\end{equation}
Thus, $\lambda_{k} = jk + n\ZZ$, for each $k$.

\bigskip

($\Longleftarrow$) 
Suppose the eigenvalues of $G$ are of the form
$\lambda_{k} = jk + c_{k}n$,
where $c_{k} \in \ZZ$ are integers and $j \in \ZZ$ is an integer with $\gcd(j,n)=1$.
To show $G$ has universal perfect state transfer, by Lemma \ref{lemma:upst-zero-to-one}, 
it suffices to show $G$ has perfect state transfer from $0$ to $1$ at some time $t$.
We have
\begin{equation}
\tbracket{\buket{1}}{e^{-itA(G)}}{\buket{0}}
	= \sum_{k=0}^{n-1} e^{-it\lambda_{k}} 	
		\bracket{\buket{1}}{\bvket{F_{k}}} \bracket{\bvket{F_{k}}}{\buket{0}} 
	= \frac{1}{n} \sum_{k=0}^{n-1} e^{-i(t\lambda_{k} - 2\pi k/n)}.
\end{equation}
Using the assumed form of the eigenvalues $\lambda_{k}$, we get
\begin{equation}
\tbracket{\buket{1}}{e^{-itA(G)}}{\buket{0}}
	= \frac{1}{n} \sum_{k=0}^{n-1} e^{-i(t(jk + c_{k}n) - 2\pi k/n)} 
	= \frac{1}{n} \sum_{k=0}^{n-1} e^{-itc_{k}n} (e^{-i(tj - 2\pi/n)})^{k}
\end{equation}
By letting $t = 2\pi m/n$, where $m$ is an integer which satisfies $jm \equiv 1\pmod{n}$, we get
\begin{equation}
\tbracket{\buket{1}}{e^{-itA(G)}}{\buket{0}}
=
\frac{1}{n} \sum_{k=0}^{n-1} e^{-2\pi ic_{k}m} (e^{-2\pi i(jm-1)/n})^{k}
=
1.
\end{equation}
This proves the claim.
\end{proof}

%%%%%%%%%%%%%%%%%%%%%%%%%%%%%%%%%%%%%%%%%%%%%%%%%%%%%%%%%%%%%%%%%%%%%%%%%%%%%%%%%%%%%%%%%%%%%%%%%%%%%%%%%%%%%%%%%%

\section{Real Symmetric Graphs}

Kay \cite{kay11} showed that universal perfect state transfer is not possible 
for graphs with real symmetric adjacency matrices. 
We state formally this important observation in the following.

\begin{theorem} (Kay \cite{kay11})
Let $G$ be a graph with real symmetric adjacency matrix.
If $G$ has perfect state transfer from vertex $a$ to vertex $b$ 
and also from vertex $a$ to vertex $c$, then we have $b = c$.
\end{theorem}

In contrast, we show that there are graphs with real symmetric adjacency matrices 
with universal pretty good state transfer. For our result, we appeal to a theorem 
of Lindemann from number theory (see Theorem 9.1 in Niven \cite{niven}). 

\begin{theorem} \label{thm:lindemann}
(Lindemann)
Given any distinct algebraic numbers $\alpha_{1},\ldots,\alpha_{m}$,
the values $e^{\alpha_{1}},\ldots,e^{\alpha_{m}}$ are linearly independent over
the field of algebraic numbers.
\end{theorem}

\begin{figure}[t]
\begin{center}
\begin{tikzpicture}[scale=4]
\foreach \c [count=\x from 0] in {{1-e-e^2+e^3},{1-e-e^2+e^3}}
{	
	\node at ({\x*180}:1.1)[circle][scale=0.75]{$\c$};
}
\foreach \c [count=\x from 0] in {{1-e+e^2-e^3}, {1-e+e^2-e^3}}
{
	\node at ({\x*180+90}:0.87)[circle][scale=0.75]{$\c$};
}
\foreach \x in {45, 135, ..., 360}
\foreach \y in {90, 180}	
{
	\draw ({\x}:1) --(\x+\y:1);
}
\foreach \c [count=\x from 0] in {0, ..., 3} %list for label, count for degree
{
	\node at ({\x*90+45}:1)[circle, circular drop shadow, fill={blue!20}, draw][scale=0.75]{$\c$};	
}
\node at (10:0.4)[circle][scale=0.75]{$1+e-e^2-e^3$};
\node at (170:0.4)[circle][scale=0.75]{$1+e-e^2+e^3$};
\end{tikzpicture}
\caption{The smallest graph with real symmetric adjacency matrix which
has universal {\em pretty good} state transfer. 
No graph with real symmetric adjacency matrix can have universal {\em perfect} 
state transfer (see Kay \cite{kay11}).}
\end{center}
\end{figure}
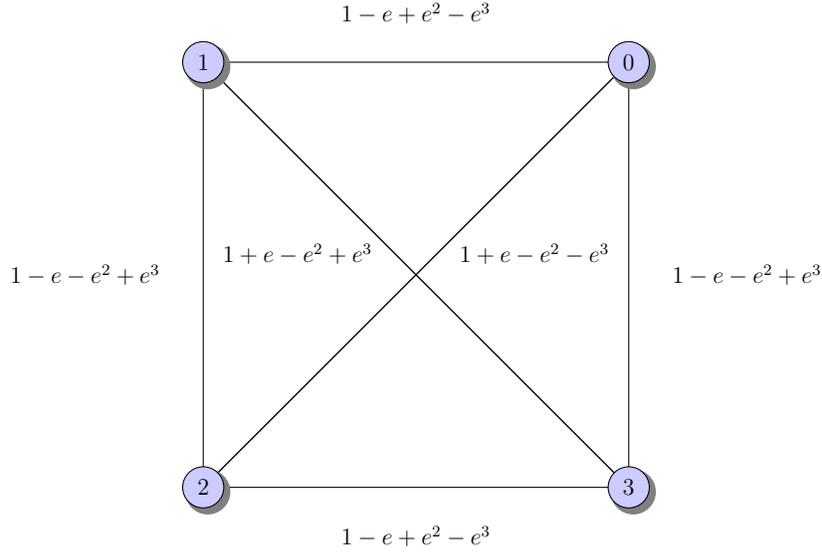

\begin{theorem}
There is a family of (weighted) graphs with universal pretty good state transfer 
whose adjacency matrices are real symmetric. 
\end{theorem}
\begin{proof}
For a positive integer $n$, let $H_{n}$ be the Hadamard matrix (of Sylvester type)
of order $N = 2^{n}$ which satisfies $H_{n}^{T}H_{n} = 2^{n}\II$. 
Recall that these Hadamard matrices may be defined recursively by
$H_{n} = H_{1} \otimes H_{n-1}$, for $n \ge 2$, with
$H_{1} = \begin{bmatrix} 1 & 1 \\ 1 & -1 \end{bmatrix}$.
In what follows, we will identify the set of integers $\{0,\ldots,2^{n}-1\}$
with the set $\zo^{n}$ of binary sequences of length $n$.

For each $z \in \zo^{n}$, we choose $\alpha_{z} = z$ be $N$ distinct (algebraic) numbers.
By Lindemann's Theorem \ref{thm:lindemann}, the set $\{e^{\alpha_{z}} : z \in \zo^{n}\}$ 
is linearly independent over the rationals.
Let $G_{n}$ be a graph on $N$ vertices whose adjacency matrix is given by 
\begin{equation}
A(G_{n}) 
= 
H_{n} 
\begin{bmatrix}
\exp(\alpha_{0}) & 0 & 0 & \ldots & 0 \\
0 & \exp(\alpha_{1}) & 0 & \ldots & 0 \\
\vdots & \vdots & \vdots & \ldots & \vdots \\
0 & 0 & 0 & \ldots & \exp(\alpha_{N-1})
\end{bmatrix}
H_{n}^{T}.
\end{equation}
We show that $G_{n}$ has universal pretty good state transfer.
Let the eigenvalues of $G_{n}$ be denoted $\lambda_{z} = e^{\alpha_{z}}$
and let the $z$th eigenvector be denoted $\ket{\chi}_{z}$ 
(which is the $z$th column of $H_{n}$). 
For each $a \in \zo^{n}$, we have
\begin{equation}
\bracket{\buket{a}}{\bvket{\chi}_{z}} = \frac{1}{\sqrt{N}} \prod_{k=1}^{n} (-1)^{a_{k}z_{k}},
\end{equation}
where $a_{k}$ and $z_{k}$ represent the $k$th bits of $a$ and $z$, respectively.
Thus, for a given vertex $a \in \zo^{n}$ of $G_{n}$, we get
\begin{equation} \label{eqn:lindemann}
\tbracket{\buket{a}}{e^{-itA(G_{n})}}{\buket{0_{n}}}
	= \sum_{z \in \zo^{n}} e^{-it\lambda_{z}}
		\bracket{\buket{a}}{\bvket{\chi}_{z}}
		\bracket{\bvket{\chi}_{z}}{\buket{0_{n}}} 
	= \frac{1}{N} \sum_{z \in \zo^{n}} e^{-it\lambda_{z}} \bracket{\buket{a}}{\bvket{\chi}_{z}}.
\end{equation}
Since the eigenvalues $\lambda_{z}$'s are linearly independent over the rationals,
by Kronecker's Theorem \ref{thm:kronecker}, 
for any $\epsilon > 0$, there is a real number $t \in \RR$ so that
\begin{equation}
\left|t\lambda_{z} - 2^{\iverson{\braket{a}{\chi_{z}} = 1}}\pi - 2\pi p_{z}\right| 
	< \epsilon,
\end{equation}
for some integers $p_{z} \in \ZZ$.
Therefore, $\tbracket{\buket{a}}{e^{-itA(G_{n})}}{\buket{0_{n}}} \approx 1$ 
in Equation (\ref{eqn:lindemann})
since $e^{-it\lambda_{z}} \approx \bracket{\buket{a}}{\bvket{\chi}_{z}}$, 
for each $z$.
This proves the claim.
\end{proof}

%%%%%%%%%%%%%%%%%%%%%%%%%%%%%%%%%%%%%%%%%%%%%%%%%%%%%%%%%%%%%%%%%%%%%%%%%%%%%%%%%%%%%%%%%%%%%%%%%%%%%%%%%%%%%%%%%%

\section{Conclusions}

In this paper, we studied the problem of {\em universal} (pretty good or perfect)
state transfer on graphs. This is a stronger and natural extension to the notion 
of pretty good and perfect state transfer which have been studied extensively in 
quantum walks on graphs.
As our main contribution, we proved spectral and structural properties of graphs 
with universal state transfer and showed several infinite family of graphs with
this property.

We showed that if $G$ is a $n$-vertex graphs with universal state transfer, 
then $G$ has $n$ distinct eigenvalues and its eigenbasis is flat. 
Furthermore, the switching automorphism group $\SwAut(G)$ is abelian and its order
$|\SwAut(G)|$ must divide $n$. On the other hand, if the universal state transfer is
{\em perfect}, then $\SwAut(G)$ is cyclic and that $|\SwAut(G)| = n$ if and only if
$G$ is switching isomorphic to a circulant. 
For graphs which are switching equivalent to circulants, we proved a spectral
characterization for universal perfect state transfer. Here, we showed that the 
eigenvalues of such a $n$-vertex graph must be the image of the integers modulo 
$n$ under a linear bijection up to time-scaling and diagonal-shifting.

We also described an infinite family of graphs (with Hermitian adjacency matrices)
with universal pretty good state transfer. These graphs are obtained from directed
prime-length cycles with $\pm i$ weights. Finally, we showed a family of graphs
with real symmetric adjacency matrices that have universal {\em pretty good} state transfer.
In contrast, as observed by Kay, no graphs with real symmetric adjacency matrices can have 
universal {\em perfect} state transfer.

It would be interesting to find a family of graphs with universal perfect state transfer. 
Currently, the only known examples are $\mathcal{K}_{2}$ and $\SC_{3}$.
Also, the graph $\mathcal{K}_{4}$ described in Section \ref{section:small-graphs} is 
the only known example of a universal pretty good state transfer graph that is not a 
circulant. It is unclear if there is a family of graphs with universal state transfer 
that generalizes $\mathcal{K}_{4}$ in a natural way. 
We leave these as open questions for future work.

%%%%%%%%%%%%%%%%%%%%%%%%%%%%%%%%%%%%%%%%%%%%%%%%%%%%%%%%%%%%%%%%%%%%%%%%%%%%%%%%%%%%%%%%%%%%%%%%%%%%%%%%%%%%%%%%%%

\section*{Acknowledgments}

The research was supported in part by the National Science Foundation grant DMS-1262737.
We thank Chris Godsil for a great many helpful advice on graph spectra and quantum walks.

%%%%%%%%%%%%%%%%%%%%%%%%%%%%%%%%%%%%%%%%%%%%%%%%%%%%%%%%%%%%%%%%%%%%%%%%%%%%%%%%%%%%%%%%%%%%%%%%%%%%%%%%%%%%%%%%%%

%\bibliography{qwalk}

%%%%%%%%%%%%%%%%%%%%%%%%%%%%%%%%%%%%%%%%%%%%%%%%%%%%%%%%%%%%%%%%%%%%%%%%%%%%%%%%%%%%%%%%%%%%%%%%%%%%%%%%%%%%%%%%%%

\newpage
\appendix

\section*{Appendix}

\begin{lemma} \label{lemma:complex-sum-forcing}
Let $\beta_{1},\ldots,\beta_{n}$ be a set of positive real numbers
which satisfies $\sum_{k=1}^{n} \beta_{k} = 1$.
Let $\alpha_{1},\ldots,\alpha_{n}$ be a set of real numbers where
$|\sum_{k=1}^{n} e^{i\alpha_{k}}\beta_{k}| = 1$.
Then, there exists $\alpha \in \RR$ so that $\alpha_{k} = \alpha$, for all $k=1,\ldots,n$.
\end{lemma}
\begin{proof}
By squaring both identities, we get
\begin{equation}
1 = \sum_{k=1}^{n} \beta_{k}^{2} + 
		\sum_{1 \le k < \ell \le n} 2\Real(e^{i(\alpha_{k}-\alpha_{\ell})}\beta_{k}\beta_{\ell})
	\le \sum_{k=1}^{n} \beta_{k}^{2} + 
		\sum_{1 \le k < \ell \le n} 2\beta_{k}\beta_{\ell}
	= 1.
\end{equation}
Here, we use $\Real(z) \le |z|$, which holds for any $z \in \CC$.
Thus, we get
\begin{equation}
\sum_{1 \le k < \ell \le n} \Real(e^{i(\alpha_{k}-\alpha_{\ell})})\beta_{k}\beta_{\ell}
= 
\sum_{1 \le k < \ell \le n} \cos(\alpha_{k}-\alpha_{\ell})\beta_{k}\beta_{\ell}
= 
\sum_{1 \le k < \ell \le n} \beta_{k}\beta_{\ell}.
\end{equation}
Since $\beta_{k} > 0$, we have $\alpha_{k} = \alpha_{\ell}$, 
for each $1 \le k < \ell \le n$.
\end{proof}

%%%%%%%%%%%%%%%%%%%%%%%%%%%%%%%%%%%%%%%%%%%%%%%%%%%%%%%%%%%%%%%%%%%%%%%%%%%%%%%%%%%%%%%%%%%%%%%%%%%%%%%%%%%%%%%%%%

\end{document}